\documentclass[11pt]{article}
\usepackage{geometry}
\usepackage{amsthm, amsmath, amssymb, bbm, bm, mathrsfs}
\usepackage[round]{natbib}
\usepackage[colorlinks, 
			linkcolor=blue,
			anchorcolor=blue,
			urlcolor=blue,
			citecolor=blue]{hyperref}
\usepackage{graphicx}
\usepackage{mathrsfs}
\usepackage{threeparttable}
\usepackage[title]{appendix}
\usepackage{url}
\graphicspath{{figure/}}

\usepackage{diagbox}
\usepackage{booktabs}
\usepackage{caption}
\usepackage{subfigure}
\usepackage{float}
\usepackage{listings}
\usepackage{multirow}
\usepackage{rotating}
\usepackage{longtable}
\usepackage{enumitem}
\usepackage{cases}
\usepackage{filecontents}
\usepackage{algorithm}
\usepackage{algorithmic}
\usepackage{bibunits}
\defaultbibliographystyle{apalike}

\defaultbibliography{reference} 

\usepackage{makecell}
\usepackage{adjustbox}
\usepackage{xr}

\newtheorem{theorem}{Theorem}[section]

\newtheorem{lemma}{Lemma}[section]

\newtheorem{assumption}{Assumption}

\newtheoremstyle{mycase}{5pt}{5pt}{\upshape}{}{\bfseries}{.}{ }{} \theoremstyle{mycase}

\newtheoremstyle{myexample}{5pt}{5pt}{\upshape}{}{\bfseries}{.}{ }{} \theoremstyle{myexample}
\newtheorem{example}{Example}

\renewcommand{\theequation}{\thesection.\arabic{equation}}
\numberwithin{equation}{section}

\renewcommand{\hat}{\widehat}
\renewcommand{\tilde}{\widetilde}

\newcommand{\Cov}{\mathrm{Cov}}
\newcommand{\Var}{\mathrm{Var}}
\newcommand{\dd}{\mathrm{d}}

\newcommand{\mbR}{\mathbb{R}}
\newcommand{\mbS}{\mathbb{S}}
\newcommand{\mbM}{\mathbb{M}}

\newcommand{\mcB}{\mathcal{B}}
\newcommand{\mcT}{\mathcal{T}}

\newcommand{\mcF}{\mathcal{F}}
\newcommand{\mcH}{\mathcal{H}}

\newcommand{\mcD}{\mathcal{D}}
\newcommand{\mcN}{\mathcal{N}}

\renewcommand{\baselinestretch}{1.1}
\begin{document}
\allowdisplaybreaks[3] 

\title{\bf\Large SPARK: A General Goodness-of-Fit Assessment via Residual Projection}
\author{{Xingwei Liu, Yuhong Yang, Wangli Xu\thanks{The corresponding author}}\\
	{\small \it Center for Applied Statistics and School of Statistics, Renmin University of China, Beijing, China}\\
    {\small \it Yau Mathematical Science Center, Tshinghua University, Beijing, China}}

\date{}
\maketitle

\vspace{-0.25in}
	
\begin{abstract}
Goodness-of-fit testing is a basic tool for assessing whether a fitted procedure has captured the systematic information contained in the covariates. 
While traditional theory has largely focused on parametric regression models, modern data analysis increasingly relies on flexible black-box learners, whose predictive success alone is insufficient to assess model accuracy. 
In this paper, we propose SPARK, a general framework for goodness-of-fit testing that applies to traditional statistical models and general black-box learning procedures, continuous and binary responses, and low- and high-dimensional predictors. 
Based on a debiasing strategy, the residuals from an initial fit of a learning procedure are projected onto nearly orthogonal directions to extract any remaining signal. 
To capture information across all projection directions, we propose a kernel-based projection method and establish both its asymptotic properties and the consistency of a bootstrap procedure. 
Comprehensive simulations and real data analyses illustrate the effectiveness and flexibility of our proposed method.

\medskip
		

\end{abstract}

\section{Introduction}

Goodness-of-fit testing is a fundamental problem in statistics and, increasingly, in machine learning. In classical statistical modeling, it safeguards inference against model misspecification, under which standard errors, confidence intervals and tests may be unreliable. 
In modern machine learning, the need for such assessment is arguably even more acute. 
Black-box procedures, including random forests \citep{Breiman2001RandomF}, gradient boosting \citep{Friedman2001GreedyFA} and neural networks \citep{Schmidhuber2015DeepLI}, are highly flexible and often achieve strong predictive performance, yet their complexity makes it difficult to determine whether the fitted learner has exhausted the predictive information in the covariates, or whether detectable structure remains in the residuals.


In the early literature, goodness-of-fit tests primarily examined whether the mean function belongs to a prescribed parametric class. One line of research, commonly referred to as global smoothing tests, constructed test statistics from empirical processes based on residuals; see \cite{Stute1997NonparametricMC,Stute1998ModelCF,Stute2008ModelDF}. Another line of goodness-of-fit tests relied on local smoothing techniques, typically estimating the conditional mean of the residuals given the predictors via nonparametric regression; see \cite{Hrdle1993ComparingNV,Zheng1996ACT,Guo2016ModelCF}. 
In particular, \cite{Fan2001GoodnessoffitTF} proposed an adaptive Neyman test to assess the goodness-of-fit for parametric regression models by testing whether the conditional mean of residuals given predictors is close to zero. 
For a recent review, we refer to \cite{Tan2025WeightedRE}.


More recently, in high-dimensional settings, \cite{Shah2018GoodnessoffitTF} introduced the residual prediction test for linear models, which fits an arbitrary prediction function to the residuals. \cite{Jankova2020GoodnessoffitTI} further extended the idea and proposed goodness-of-fit tests for generalized linear models. Using data splitting together with a debiasing strategy based on the square-root LASSO, they applied nonparametric techniques to predict the residual terms and constructed a direction orthogonal to the residuals, thereby assessing whether additional information could be captured. 
Moreover, several recent efforts have sought to extend traditional goodness-of-fit tests to black-box learners, a challenging problem due to the complexity of the underlying learning procedures. \cite{Zhang2023IsAC} and \cite{Javanmard2024GRASPAG} proposed unified frameworks for conducting goodness-of-fit tests for binary black-box classifiers. \cite{He2025AGA} generalized the framework of \cite{Zhang2023IsAC} to accommodate both binary and continuous responses in high-dimensional settings. 

In this paper, we propose a novel goodness-of-fit test, named the statistical projection assessment with reproducing kernels (SPARK), for both traditional model-based learners and general black-box learning procedures. It is constructed based on data splitting, where one subset is used to train the interested regression model and the other is used for evaluating the goodness-of-fit. Drawing on the spirit of \cite{Jankova2020GoodnessoffitTI}, we project the residuals orthogonally via a debiasing strategy, which can be regarded a degenerate case of double machine learning \citep{Chernozhukov2018DoubleDebiasedML}. 
Our work contributes in three aspects to the literature: 
(a) We develop a unified framework for conducting goodness-of-fit tests that is applicable to both classic model-based learners and general black-box learners. In contrast,  \cite{Shah2018GoodnessoffitTF} and \cite{Jankova2020GoodnessoffitTI} are designed for checking parametric models. 
(b) Our method accommodates both regression and binary classification tasks, in both low- and high-dimensional regions. 
\cite{He2025AGA} is limited to high-dimensional settings, while \cite{Zhang2023IsAC} and \cite{Javanmard2024GRASPAG} are restricted to binary classifiers. 
(c) From the perspective of projection, the SPARK can fully characterize the null hypothesis by considering all the projection directions, whereas \cite{He2025AGA} took several special projection directions. 
Comprehensive numerical studies illustrate the superior empirical performance of our proposal in terms of both Type-I error rate and statistical power. 
 
The rest of the paper is organized as follows. 
In Section~\ref{sec_02}, we formulate the testing problem and present a simple testing procedure to illustrate the main idea. 
The SPARK procedure is formally presented in Section~\ref{sec_03}. 
Section~\ref{sec_04} includes various topics related to implementing the proposed test. 
We present simulation results and real-data examples in Sections~\ref{sec_05} and~\ref{sec_06}, respectively. 
Section~\ref{sec_07} concludes the paper with a brief discussion. 
Technical proofs and additional simulations are included in the supplementary material. 

\textit{Notation.}
For a probability measure $\mu$ on $\mathbb{R}^d$ and $1 \le p \le \infty$, let $L_p(\mu) = \{ g : \mathbb{R}^d \to \mathbb{R} \mid \|g\|_{L_p(\mu)} < \infty \}$, where $\|g\|_{L_p(\mu)}$ denotes the usual $L_p$-norm. In particular, $L_\infty(\mu)$ denotes the space of essentially bounded functions, and $\|g\|_{L_\infty(\mu)}$ denotes the essential supremum with respect to $\mu$. 
For two sequences $a_n$ and $b_n$, the notation $a_n = O(b_n)$ indicates that there exists a constant $C>0$ such that $|a_n| \leq C|b_n|$ for all sufficiently large $n$. 
$a_n = o(b_n)$ denotes that $a_n/b_n \to 0$ as $n \to \infty$. 
For a random sequence $A_n$, the notation $A_n = O_p(b_n)$ means that $A_n/b_n$ is bounded in probability, and $A_n = \Omega_p(b_n)$ means that $b_n = O_p(A_n)$.


\section{Methodology}\label{sec_02}
\subsection{Problem Formulation}
Consider a general condition mean regression model as follows:
\begin{equation}\label{eq_2_01}
    Y=E(Y|X)+\epsilon =: m(X) + \epsilon,
\end{equation}
where $Y\in\mbR$ denotes the response and $X\in\mbS\subseteq\mbR^p$ is the vector of predictors. To facilitate theoretical derivations, we require a mild moment condition on model \eqref{eq_2_01} as stated in Assumption~\ref{ass_01}.
\begin{assumption}\label{ass_01}
For $C>0$, $0<E(\epsilon^2|X)\leq C$ a.e.
\end{assumption}
Except for the moment condition specified in Assumption~\ref{ass_01}, we impose no distributional assumption on the regression error on this model. Therefore, it encompasses a broad class of homoscedastic and heteroscedastic regression models that target the conditional mean of a scalar response and whose conditional variances satisfy Assumption~\ref{ass_01}. 
Examples include the linear regression model with $m(X)=\beta_0+X^\top\beta$,
generalized linear models with $m(X)=g^{-1}(\beta_0+X^\top\beta)$ for a link function $g(\cdot)$ \citep{McCullaghNelder1989}, and semiparametric models such as the partially linear model $m(X)=Z^\top\beta+h(W)$ with $X^\top=(Z^\top,W^\top)$ \citep{Robinson1988} and the single-index model $m(X)=h(X^\top\beta)$ for unknown function $h(\cdot)$ \citep{Ichimura1993}. The model also includes fully nonparametric regression models with an unspecified mean function. 

Let $\mcD_n=\{X_i,Y_i\}_{i=1}^{n}$ be $n$ independent observations from the joint distribution. 
Using the observed sample and a learning procedure of interest, we obtain $m_{\mcD_n}(x)$ as an estimate of the conditional expectation $m(x)$. The regression model~\eqref{eq_2_01} can be rewritten as 
\begin{equation*}
    Y=m_{\mcD_n}(X) + \varepsilon_n,
\end{equation*}
where $\varepsilon_n=Y-m_{\mcD_n}(X)= \epsilon + m(X) - m_{\mcD_n}(X)$ represents the residual term. 
From the perspective of controlling generalization errors, \cite{He2025AGA} considered to assess the learning procedure by quantifying $\sup_{x\in\mbS}|m_{\mcD_n}(x)-m(x)|$. 
Let $r_n=o(1)$ be the convergence rate of the learning procedure we assess under the null hypothesis. 
The hypotheses of interest can be formulated as 
\begin{equation}\label{eq_2_03}
\begin{aligned}
    H_0: &\exists\mbM_n^0\subseteq\mbS\text{ with }P(X\in\mbM_n^0)=1 \text{ such that }\sup_{x\in\mbM_n^0}|m_{\mcD_n}(x)-m(x)|=O_p(r_n),\\
    H_1: &\exists\mbM_n^1\subseteq\mbS \text{ with } P(X\in\mbM_n^1)\geq c\text{ such that }\inf_{x\in\mbM_n^1}|m_{\mcD_n}(x)-m(x)|=\Omega_p(r_n^{(a)}),
\end{aligned}
\end{equation}
where $c$ is some positive constant and $r_n^{(a)}$ is a sequence such that $r_n=o(r_n^{(a)})$. In other words, $r_n^{(a)}$ may either decay at a slower rate than $r_n$ or fail to converge to zero, indicating that the learning procedure converges more slowly or not at all to $m(X)$, respectively. 

We present several statistical examples relevant to hypotheses in \eqref{eq_2_03}. 
\begin{example}[Parametric models]
    Under suitable conditions, especially on signal strength, SCAD \citep{Fan2001VariableSV} achieves the convergence rate $\sup_{x\in\mbS}|m_{\mcD_n}(x)-m(x)|\leq sn^{-1/2}$ with high probability, where $s$ denotes the number of nonzero coefficients in the high-dimensional linear model \citep{Shi2019LINEARHT}. The null hypothesis holds by setting $r_n=sn^{-1/2}$. 
    Similarly, $r_n$ for LASSO \citep{Tibshirani1996RegressionSA} can be set as $s(\log p/n)^{1/2}$. 
\end{example}

\begin{example}[Black-box learners]
    Under suitable conditions, kernel regression \cite{Hrdle1988StrongUC} and deep learning regression \citep{Imaizumi2023SupNormCO} attain the (near) minimax rate $(\log n/n)^{\beta/(2\beta+p)}$ over the H\"older class, where $\beta$ is the smoothness parameter. The null hypothesis holds by setting $r_n=(\log n/n)^{\beta/(2\beta+p)}$. 
\end{example}

\begin{example}[Application to regression-assisted inference]
    Regression-assisted inference typically requires assumptions on the convergence rate of the regression methods, which are difficult to verify within the classical testing framework. Under the null hypothesis in \eqref{eq_2_03} and the uniform integrability condition on $\mathcal{D}_n$, the $L_2$ convergence rate satisfies that $E\!\left\{(m_{\mathcal{D}_n}(X) - m(X))^2\right\} = O(r_n^2)$, an assumption adopted, for example, in \cite{Chernozhukov2018DoubleDebiasedML} and \cite{Cai2025}. 
\end{example}

\subsection{Motivation: Starting from a Simple Statistic}\label{sec_22}
To facilitate our proposal, we reformulate $H_0$ based on Lemma~\ref{lemma_01}. 
Denote $\mcF_n:=\{f_{\mcD_n}\mid f_{\mcD_n}:\mbS\to \mbR, \|f_{\mcD_n}\|_{L_1(P_X)}\leq1\text{ for }\mcD_n\ a.e.\}$.

\begin{lemma}\label{lemma_01}
    Suppose that $\sup_{f_{\mcD_n}\in\mcF_n}E\{|\epsilon f_{\mcD_n}(X)|\,|\mcD_n\}<\infty$ and $\sup_{n\geq1}|m_{\mcD_n}-m|\in L_{\infty}(P_X)$ a.e.\ hold.
    Then the null hypothesis in \eqref{eq_2_03} is equivalent to
    $$
    \sup_{f_{\mcD_n}\in\mcF_n}|E\{\varepsilon_n f_{\mcD_n}(X)|\mcD_n\}|=O_p(r_n).
    $$
    Moreover, a necessary condition of the null hypothesis is that 
    \begin{equation}\label{eq_2_04}
    \sup_{f_{\mcD_n}\in\mcF_n}|E[\varepsilon_n \{f_{\mcD_n}(X)-E(f_{\mcD_n}(X)|\mcD_n)\}|\mcD_n]|=O_p(r_n).
    \end{equation}
\end{lemma}

It is worth noting that \eqref{eq_2_04} is equivalent to the null hypothesis under the additional assumption that $E(\varepsilon_n|\mcD_n)E(f_{\mcD_n}(X)|\mcD_{n})=O_p(r_n)$, which particularly holds if $E(\varepsilon_n|\mcD_n)=O_p(r_n)$, meaning that any systematic shift in the error conditional on the training data is of at most order $r_n$. 
The additional centering term is crucial for debiasing, which can be regarded as a degenerate case of the conditional mean in double machine learning \citep{Chernozhukov2018DoubleDebiasedML}. It is expected that the centered projection directions are nearly orthogonal to the residuals from an asymptotic perspective. 
Moreover, our methodology is similar to that of \cite{Jankova2020GoodnessoffitTI}, but is more general. Specifically, they projected the residuals orthogonally using the square-root LASSO to reduce biases in the context of generalized linear models, whereas our methodology is applicable for both parametric models and black-box learning procedures.

Motivated by Lemma~\ref{lemma_01}, we consider splitting the dataset $\mcD_n$ into two subsets, a training subset and a testing subset, denoted as $\mcD_{n_1}:=\{X_{1,i},Y_{1,i}\}_{i=1}^{n_1}$ and $\mcD_{n_2}:=\{X_{2,i},Y_{2,i}\}_{i=1}^{n_2}$ with $n_1+n_2=n$, respectively. Specifically, $\mcD_{n_1}$ is used to estimate the conditional expectation $m(X)$ using the learning procedure of interest, denoted as $m_{\mcD_{n_1}}(X)$. 
$\mcD_{n_2}$ is used to calculate the residuals, which are given by $\varepsilon_{n_{2,i}}=Y_{2,i}-m_{\mcD_{n_1}}(X_{2,i})$. 
Given $\mcD_{n_1}$, $f_{\mcD_{n_1}}(\cdot)$ in \eqref{eq_2_04} is a fixed function $f(\cdot)$ defined on $\mbR^p$ satisfying $\|f\|_{L_1(P_X)}\leq1$.  
To estimate $E[\varepsilon_n \{f(X)-E(f(X)|\mcD_n)\}|\mcD_n]$, we can construct the following statistic
\begin{equation}\label{eq_2
_05}
    T_f:=\frac{1}{n_2}\sum_{i=1}^{n_2}\{Y_{2,i}-m_{\mcD_{n_1}}(X_{2,i})\}\{f(X_{2,i})-\hat m_f\},
\end{equation}
where  $\hat m_f:= \sum_{i=1}^{n_1}f(X_{1,i})/n_1$. 
Let $f$ be any function such that $\Var(f(X))<\infty$ and $0<V_f<\infty$, where $V_f:=E[\{Y-m(X)\}^2\{f(X)-E(f(X))\}^2]$. The following theorem establishes the limiting distribution under the null hypothesis for the statistic constructed with such a choice of $f$.

\begin{theorem}\label{th_01}
    Suppose that Assumption~\ref{ass_01} holds. 
    Under the null $H_0$, if $n_1,n_2\to\infty$ and $n_2=o(r_{n_1}^{-2})$, it holds that
    $$
    \sqrt{n_2}T_f\overset{d}{\to}\mcN(0,V_f),
    $$
    where $\overset{d}{\to}$ denotes convergence in distribution. 
\end{theorem}


Theorem~\ref{th_01} demonstrates that the size of $T_f$ can be asymptotically controlled when $n_2=o(r_{n_1}^{-2})$. However, its statistical power is not guaranteed in general. Intuitively, when $m(\cdot)-m_{\mcD_{n_1}}(\cdot)$ is nearly orthogonal to the projection direction $f(\cdot)$ with respect to $L_2(P_X)$ metric, the test statistic $T_f$ may fail to detect the signal under the alternatives, even if $\inf_x|m(x)-m_{\mcD_{n_1}}(x)|$ is large. To state this point more rigorously, we establish the asymptotic distribution of $T_f$ under the alternatives subject to additional Lyapounov conditions; see Appendix~\ref{app_sec_04} for details.

\section{Kernel-Based Projection}\label{sec_03}
As discussed in the last section, a general projection direction $f(\cdot)$ may suffer from an inconsistent power. Moreover, although Lemma~\ref{lemma_01} provides an equivalent characterization of the null hypothesis, it is not directly applicable in practice since it is infeasible to account for all functions in $\mcF_n$. In this section, we propose the SPARK that considers all the projection directions within a dense subset of $L_1(P_X)$. 

\subsection{Kernel-Based Test Statistic}
Let $k:\mbS\times\mbS\to\mbR$ denote a symmetric positive-definite kernel, and $\mcH_k$ denote the reproducing kernel Hilbert space (RKHS) associated with $k$. 
Denote $\mcH_k^0:=\{f_{\mcD_n}\mid f_{\mcD_n}:\mbS\to \mbR, f_{\mcD_n}\in\mcH_k,\|f_{\mcD_n}\|_{\mcH_k}\leq1 \text{ for every given }\mcD_n\}$. 
Motivated by Lemma~\ref{lemma_01}, we consider a kernel embedding as
\begin{equation}\label{eq_fk}
    \sup_{f_{\mcD_n}\in\mcH_k^0}E\big[\varepsilon_n\{f_{\mcD_n}(X)-E(f_{\mcD_n}(X)|\mcD_{n})\}|\mcD_{n}\big].
\end{equation}
The following lemma illustrates that taking the supremum over the subset $\mcH_k^0$ is equivalent to taking it over the whole set. The condition on the kernel particularly holds when the kernel $k$ is universal \cite[Definition~4]{Steinwart2001OnTI}, that is, $\mcH_k$ is dense in the Banach space of bounded continuous functions with respect to the supremum norm. Examples include the Gaussian kernel and the Laplacian kernel; see details in \cite{Sriperumbudur2010HilbertSE}. 
\begin{lemma}\label{lemma_02}
    Suppose $\mbS$ is a compact subset of $\mbR^p$ and conditions in Lemma~\ref{lemma_01} hold. Suppose $k$ is a continuous symmetric positive-definite kernel function with $k(x,x)>0$ for any $x\in\mbS$, and $\mcH_k$ is dense in $L_1(P_X)$ with respect to $L_1$-norm. 
    Then, \eqref{eq_2_04} holds if and only if 
    $$\sup_{f_{\mcD_n}\in\mcH_k^0}|E\big[\varepsilon_n\{f_{\mcD_n}(X)-E(f_{\mcD_n}(X)|\mcD_n)\}|\mcD_{n}\big]|=O_p(r_n).$$
\end{lemma}

Thanks to the reproducing property, taking square on both sides of \eqref{eq_fk} yields a closed form as 
\begin{equation}\label{eq_fk2}
    E\big[\varepsilon_n\varepsilon_n'U_n(X,X')|\mcD_{n}\big],  
\end{equation}
where 
$U_n(X,X'):= k(X,X')-E_X(k(X,X')|\mcD_{n})-E_{X'}(k(X,X')|\mcD_{n})+E_{XX'}(k(X,X')|\mcD_{n})$ and 
$(X',\varepsilon_n')$ is an independent copy of $(X,\varepsilon_n)$. 

Drawing on the spirit in Section~\ref{sec_22}, we split the dataset $\mcD_n$ into a training subset $\mcD_{n_1}$ and a testing subset $\mcD_{n_2}$ and then estimate \eqref{eq_fk2} as 
\begin{equation*}
    T_k:=\frac{1}{n_2(n_2-1)}\sum_{i\neq j}^{n_2}\{Y_{2,i}-m_{\mcD_{n_1}}(X_{2,i})\}\{Y_{2,j}-m_{\mcD_{n_1}}(X_{2,j})\}\hat U(X_{2,i},X_{2,j}),
\end{equation*}
where 
$$
\hat U(X_{2,i},X_{2,j})=k(X_{2,i},X_{2,j})-\frac{1}{n_1}\sum_{l=1}^{n_1}k(X_{2,i},X_{1,l})-\frac{1}{n_1}\sum_{m=1}^{n_1}k(X_{1,m},X_{2,j})+\frac{1}{n_1(n_1-1)}\sum_{l\neq m}^{n_1}k(X_{1,m},X_{1,l}). 
$$

To end this subsection, we discuss the statistic proposed in \cite{He2025AGA}, which is performed based on the cumulative covariance \citep{Zhou2020ModelFreeFS}. 
The cumulative covariance essentially measures $\Cov\{\varepsilon_n,\mathbbm1(x_k<x_0)\}$ for all $x_0$ on the support, and then summarizes the $p$ marginal effects, where $x_k$ denotes the $k$-th component of $X$. 
Consequently, their test can be interpreted as taking the $p$ marginal distribution functions as projection directions $f(\cdot)$.  
However, the connection between this particular finite choice of projection directions and projections over the entire class $\mcF_n$ has not been illustrated. It is important because Lemma~\ref{lemma_01} requires the desired property to hold for all projection directions in $\mcF_n$. Moreover, they employed a fifth-order U-statistic to establish the limiting distribution when the dimension $p$ diverges. 
In contrast, the SPARK fully characterizes the projections over the entire set under mild conditions, while requiring only a second-order U-statistic. 

\subsection{Asymptotic Properties}
We present the asymptotic properties of $T_k$ under the null and alternatives in this subsection. Denote $U(x,x'):= k(x,x')-E_X(k(X,x'))-E_{X'}(k(x,X'))+E_{XX'}(k(X,X'))$, and $\Delta_{n_1}(\cdot)=m(\cdot)-m_{\mcD_{n_1}}(\cdot)$. The following assumption entails that the mean squared error of the training model is stochastically bounded conditionally on the training data. Similar assumptions can be found in semiparametric literature \citep{Chernozhukov2018DoubleDebiasedML}. 

\begin{assumption}\label{ass_02}
    $E\{\Delta_{n_1}(X)^2|\mcD_{n_1}\}$ and $E\{\Delta_{n_1}(X)^2k(X,X)|\mcD_{n_1}\}$ are well-defined and bounded in probability as $n_1\to\infty$. 
\end{assumption}

By classical theory for the $U$-statistic, the following theorem establishes the asymptotic distribution of $n_2T_k$ under the null. 

\begin{theorem}\label{th_02}
    Suppose Assumptions~\ref{ass_01}--\ref{ass_02} hold and $E(k(X,X))<\infty$. Under the null $H_0$, if $n_1,n_2\to\infty$ and $n_2=o(r_{n_1}^{-2})$, we have 
    $$
    n_2T_k \overset{d}{\to}\sum_{r=1}^\infty\lambda_r(Z_r^2-1),
    $$
    where $Z_r$ are independent standard Gaussian random variables, and $\lambda_r$ are eigenvalues of the compact self-adjoint operator on $L_2(P_{XY})$ induced by the kernel $\{y-m(x)\}\{y'-m(x')\}U(x,x')$. 
\end{theorem}

    Next, we introduce an assumption to investigate the asymptotic property under the alternatives. Suppose $E(k(X,X))<\infty$ and $k$ is characteristic \citep{Sriperumbudur2010HilbertSE}, it can be verified that the operator $\mcT_U$ induced by the kernel $U(x,x')$ defined on $L_2(P_X)$ is a compact self-adjoint Hilbert-Schmidt operator. Then, the spectral decomposition is given by
    $$
    U(x,x') = \sum_{j=1}^{\infty} \lambda_{Uj}\phi_{Uj}(x)\phi_{Uj}(x'),
    $$
    where $\{\phi_{Uj}\}_{j=1}^{\infty}$ is an orthonormal basis of the orthogonal complement of the null space of $\mcT_U$, and eigenvalues $\lambda_{U1}\geq\lambda_{U2}\geq\dots>0$. 
    Denote $b_{n_1}=E(\Delta_{n_1}(X)|\mcD_{n_1})$ and $\Delta_{n_1}^c(\cdot) = \Delta_{n_1}(\cdot)-b_{n_1}$. 
    The centered error $\Delta_{n_1}^c$ has the orthogonal decomposition
$$
\Delta_{n_1}^c(\cdot)=\sum_{j\geq1}a_{j,n_1}\phi_{Uj}(\cdot),\quad
  a_{j,n_1}=E\{\Delta_{n_1}^c(X)\phi_{Uj}(X)|\mcD_{n_1}\}.
$$

\begin{assumption}\label{ass_03}
    There exist sequences $J_{n_1}\in\mathbb N$ and $\rho_{n_1}\in[0,1)$, possibly depending on $\mcD_{n_1}$, such that
\begin{equation}\label{eq_ass_031}
    \sum_{j=1}^{J_{n_1}}a_{j,n_1}^2\geq (1-\rho_{n_1})E\{\Delta_{n_1}(X)^2|\mcD_{n_1}\}
\end{equation}
with probability tending to $1$, and 
\begin{equation}\label{eq_ass_032}
    n_2\lambda_{UJ_{n_1}}(1-\rho_{n_1})E\{\Delta_{n_1}(X)^2|\mcD_{n_1}\}\overset{p}{\to}\infty,
\end{equation}
where the probability is over the training sample.
\end{assumption}

Assumption~\ref{ass_03} requires $m(\cdot)-m_{\mcD_{n_1}}(\cdot)$ to be well approximated by the space spanned by the first $J_{n_1}$ eigenfunctions of the integral operator induced by $U(x,x')$, in relative $L_2(P_X)$ distance. 
More precisely, since $E\{\Delta_{n_1}(X)^2|\mcD_{n_1}\}=b_{n_1}^2+\sum_{j\geq1}a_{j,n_1}^2$, the inequality \eqref{eq_ass_031} controls both $b_{n_1}$ and the portion of $\Delta_{n_1}$ lying in the spaces associated with small eigenvalues. 
The condition \eqref{eq_ass_032} requires the kernel-weighted signal contained in these leading eigenfunctions to be sufficiently strong. This condition characterizes the alternatives since $E\{\Delta_{n_1}(X)^2|\mcD_{n_1}\}=\Omega_p({r_{n_1}^{(a)}}^{2})$ under $H_1$. 
For example, if $\Delta_{n_1}(x)=r_{n_1}^{(a)}g(x)$ for some fixed $g\in L_2(P_X)$ with $E\{g(X)\}=0$, then $J_{n_1}$ and $\rho_{n_1}$ can be chosen fixed, and Assumption~\ref{ass_03} reduces to $n_2{r_{n_1}^{(a)}}^{2}\to\infty$ under $H_1$. 
Moreover, since $r_{n_1}/r_{n_1}^{(a)}\to0$ as $n_1\to\infty$, one can distinguish the null hypothesis and the alternatives with a splitting ratio such that $n_2=o(r_{n_1}^{-2})$ and $n_2=\omega\big({r_{n_1}^{(a)}}^{-2}\big)$. The following theorem presents the property of $T_k$ under the alternatives.


\begin{theorem}\label{th_power}
    Suppose that Assumptions~\ref{ass_01}--\ref{ass_03} hold and $E(k(X,X))<\infty$ with a characteristic kernel $k$. If $n_1,n_2\to\infty$, $n_1=O(n_2)$ and $b_{n_1}=O_p(1)$, then $n_2T_k$ diverges to infinity in probability. 
\end{theorem}

\subsection{Wild Bootstrap Calibration}
As illustrated in Theorem~\ref{th_02}, the limiting null distribution of $n_2T_k$ is not pivotal and creates extra difficulty on estimating unknown parameters. To construct a valid test for the proposed statistic, we adopt a wild bootstrap procedure \citep{Chernozhukov2013Gaussian} to approximate the null distribution of $n_2T_k$.
Specifically, we use $B$ to denote the bootstrap sample size, and for each $b=1,2,\dots,B$, we generate $n_2$ i.i.d.\ random multipliers $\{e_{bi}\}_{i=1}^{n_2}$ from Rademacher distribution
independent of data. The bootstrap version of $T_k$ is defined as 
\begin{equation*}
    T_k^b:=\frac{1}{n_2(n_2-1)}\sum_{i\neq j}^{n_2}\{Y_{2,i}-m_{\mcD_{n_1}}(X_{2,i})\}\{Y_{2,j}-m_{\mcD_{n_1}}(X_{2,j})\}\hat U(X_{2,i},X_{2,j})e_{bi}e_{bj}.
\end{equation*}
We reject $H_0$ at level $\alpha$ if $B^{-1}\sum_{b=1}^B\mathbbm{1}(T_k^b>T_k)<\alpha$. Let $P^*(\cdot):=P(\cdot|\mcD_n)$ denote the bootstrap probability conditional on $\mcD_{n}$, and $\overset{D^*}{\to}$ denote convergence in distribution in probability under $P^*$.
The following theorem illustrates bootstrap consistency under $H_0$. 

\begin{theorem}\label{th_03}
    Suppose Assumptions~\ref{ass_01}--\ref{ass_02} hold and $E(k(X,X))<\infty$. Under the null $H_0$, as $n_1,n_2\to\infty$, we have 
    $$
    n_2T_k^b \overset{D^*}{\to}\sum_{r=1}^\infty\lambda_r(Z_r^2-1),
    $$
    where $\lambda_r,Z_r$ are defined in Theorem~\ref{th_02}.  
\end{theorem}

The following theorem along with Theorem~\ref{th_power} implies the consistent power of the test using bootstrap.

\begin{theorem}\label{th_04}
    Suppose Assumptions~\ref{ass_01}--\ref{ass_02} hold and $E(k(X,X))<\infty$. Under the alternatives, as $n_1,n_2\to\infty$, we have $n_2T_k^b=O_{P^*}(1)$. That is, for every $\delta>0$, 
    $$\lim_{M\to\infty}\limsup_{n_1,n_2\to\infty}P(P^*(|n_2T_k^b|>M)>\delta)=0.$$
    Suppose the conditions in Theorem~\ref{th_power} also hold, then $P(n_2T_k\geq T_{k,\alpha}^*)\to1$ with $T_{k,\alpha}^*$ denoting the upper $\alpha$-quantile of $n_2T_k^b$ conditioning on the data. 
\end{theorem}

It is remarked that our main theoretical development here focuses on fixed-dimensional predictors. Our methodology also allows the dimension of the predictors to diverge, subject to several conditions required in high-dimensional regimes. The corresponding asymptotic properties and consistency of the bootstrap are established in Appendix~\ref{app_sec_E} and empirically validated in Simulations~1, 4, and 5 in Section~\ref{sec_05}.

\section{Implementation}\label{sec_04}
\subsection{Splitting Ratios in Assessing General Learning Procedures}\label{sec_41}
In this subsection, we consider different splitting ratios to assess the goodness-of-fit of the learning procedure. 
Typically, if we reject the null with a larger train set, it is more confident to regard the learning procedure as inappropriate. 
By Theorems~\ref{th_02}--\ref{th_04}, it entails that the splitting ratio satisfying $n_2=o(r_{n_1}^{-2})$ and $n_2=\omega\big({r_{n_1}^{(a)}}^{-2}\big)$ to guarantee the testing size and power. Consequently, the size of training subset could be less than the testing subset. 
Inspired by \cite{Zhang2023IsAC}, we consider three splitting ratios, where the training set size is $25\%$, $50\%$, and $75\%$, and four patterns of the assessment result are as follows. 
\begin{itemize}
    \item \textit{Pattern 1.} The test fails to reject $H_0$ under all splitting ratios, which means the learning procedure converges fast to the underlying $m(X)$.
    \item \textit{Pattern 2.} The test rejects $H_0$ only at $25\%$ training sample, which means the learning procedure converges moderately fast to the true $m(X)$. 
    \item \textit{Pattern 3.} The test rejects $H_0$ at both $25\%$ and $50\%$ and fails to reject at $75\%$, which means the learning procedure converges to $m(X)$ at a relatively slow rate. 
    \item \textit{Pattern 4.} The test rejects $H_0$ under all splitting ratios, which means the learning procedure fails to converge to the true $m(X)$. 
\end{itemize}
Practically, the practitioner could deploy more kinds of splitting ratios to evaluate the convergence of the learning procedure comprehensively.

\subsection{Stabilizing Algorithm by Multiple Splitting}\label{sec_42}

A single data splitting may cause power loss and fluctuations for the testing results. To mitigate the shortcoming, we consider a multiple-splitting procedure based on cross-validation. Specifically, the dataset is divided into $K$ folds, with $K-1$ folds used for testing and the remaining fold used for training. 
For the splitting ratios of $25\%$, $50\%$, and $75\%$ discussed in Section~\ref{sec_41}, the dataset can be partitioned into $4$, $2$, and $4$ folds, respectively. 
Once $K$ $p$-values are calculated based on the $K$-fold cross-validation, we could aggregate the result by $p$-value combination under arbitrary dependence structures \citep{Liu2020CauchyCT,Ouyang2024EffectivePC}.

\section{Simulations}\label{sec_05}
In this section, we present comprehensive numerical studies to evaluate the finite-sample performance of our proposed method. Specifically, we compare our method with existing methods associated with goodness-of-fit test in Section~\ref{sec_51}, while we use it to assess general learning procedures in Section~\ref{sec_52}.

\subsection{Comparison with Existing Methods}\label{sec_51}
In this subsection, we consider a general goodness-of-fit test for high-dimensional predictors (HCZ) proposed by \cite{He2025AGA} as a competitive method. 
For the continuous responses in Simulations~1--2, we also consider two parametric goodness-of-fit tests for linear models: the adaptive Neyman (AN) test \citep{Fan2001GoodnessoffitTF} and the residual prediction (RP) test proposed by \cite{Shah2018GoodnessoffitTF}. 
We use the default settings of RP from its corresponding R package \texttt{RPtests}. 
For the binary responses in Simulation~3, we also compare our method with BAGofT \citep{Zhang2023IsAC} implemented with the default settings of its R package \texttt{BAGofT} and GRASP with parameters $\tau=0$ and $L=10$ \citep{Javanmard2024GRASPAG}.

\textit{Simulation~1 (High-dimensional model).} Let $X=(x_1,\dots,x_p)^\top\sim\mcN(0,\Sigma)$ denote a $p$-dimensional predictor. The response variable $Y$ is generated as
$$
Y= x_1+2x_2+\exp\{\xi x_3x_4\}+\epsilon, 
$$
where the error term $\epsilon\sim\mcN(0,0.1^2)$, and $\xi\in\mbR$ is a tuning parameter. 
We fix the dimension $p=50$ and vary the sample size $n$ over $\{100,200,500,1000\}$. 
Two settings are considered for the covariance matrix: identity matrix with $\Sigma=I_p$ and correlated matrix with $\Sigma=(0.2^{|i-j|})_{p\times p}$. 

We set $\xi=0$ to examine the Type-I error and $\xi=0.5,1$ to evaluate power,  
and we consider SCAD as the interested learning procedure used for the training subset. 
As presented in Example~1, SCAD achieves a convergence rate of $(s/n_1)^{1/2}$ under the null hypothesis. Since the AN test is designed for low-dimensional settings, we provide it with oracle information by fitting least squares with covariates $(x_1,\dots,x_4)$. 
To ensure a fair comparison, we adopt a single data-splitting scheme for the tests. The splitting ratios are fixed at $95\%$ throughout the numerical studies for the HCZ test, while are set at $75\%,70\%,65\%,$ and $60\%$ for $n=100,200,500$, and $1000$ for our test, respectively. 
We consider the Gaussian kernel with bandwidth chosen according to the mean heuristic \citep{Gretton2012AKT}. 
We remark that our test is not overly sensitivity to the number of bootstrap $B$, so we simply fix $B=500$. 

\begin{figure}[htbp]\centering
\includegraphics[width=1\textwidth]{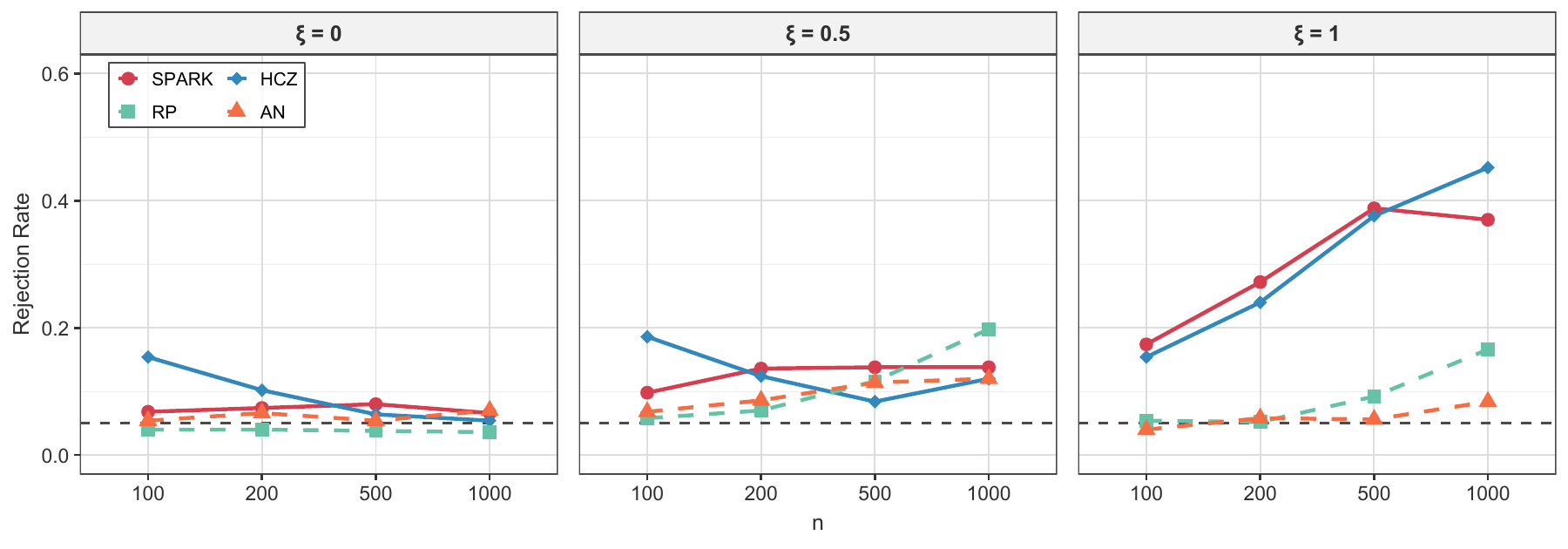}
\caption{Empirical Type-I error and power under Simulation~1 with independent predictors. The horizontal axis denotes the sample size $n$, and the vertical axis shows the rejection rate. The black dashed line indicates the nominal significance level of $0.05$. }\label{figure_1}
\end{figure}

\textit{Simulation~2 (Low-dimensional model).} Let $X=(x_1,x_2)^\top\sim\mcN(0,\Sigma)$ denote a $2$-dimensional predictor. The response variable $Y$ is generated as
$$
Y= x_1+2x_2+\exp\{\xi x_1x_2\}+\epsilon. 
$$
The linear model using ordinary least squares serves as the interested learning procedure. For the proposed test, we fix the splitting ratio at $85\%$ as a conservative choice to stabilize the finite sample performance. 
The other settings keep the same as Simulation~1. 


\begin{figure}[htbp]\centering
\includegraphics[width=1\textwidth]{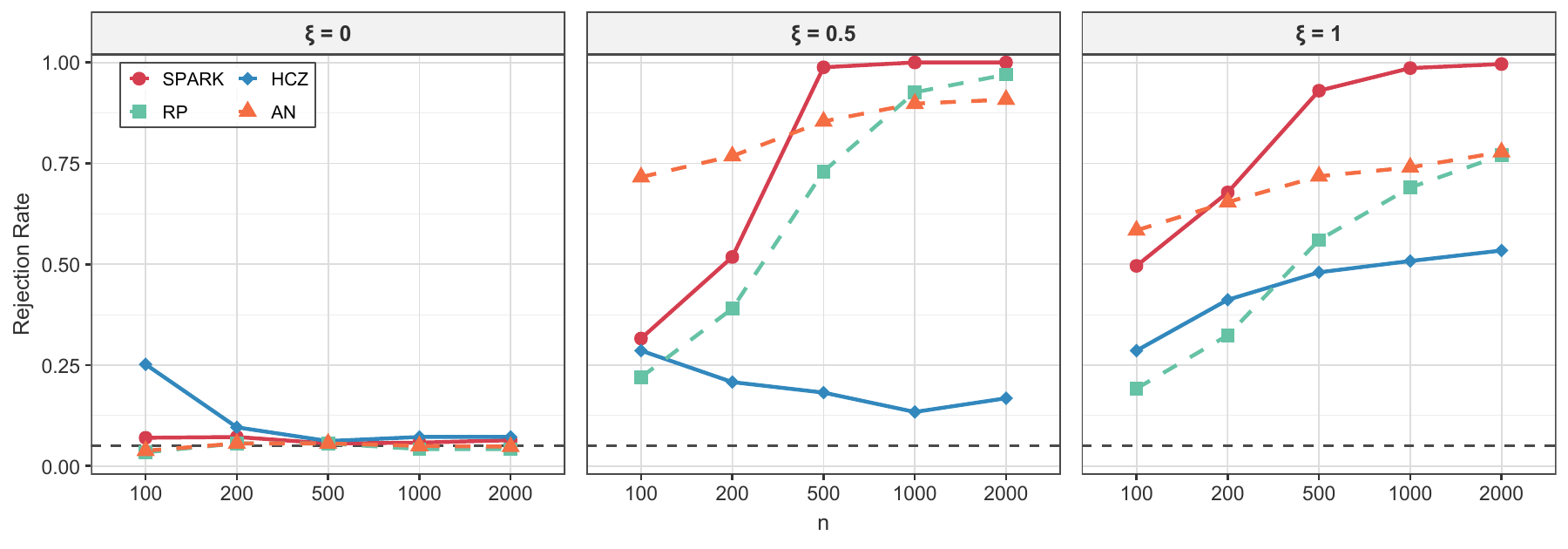}
\caption{Empirical Type-I error and power under Simulation~2 with independent predictors. The horizontal axis denotes the sample size $n$, and the vertical axis shows the rejection rate. The black dashed line indicates the nominal significance level of $0.05$. }\label{figure_2}
\end{figure}

\textit{Simulation~3 (Binary responses).} 
Let $X=(x_1,x_2,x_3)^\top\sim\mcN(0,\Sigma)$ denote a $3$-dimensional predictor. The responses are generated by 
$$
P(Y=1|x_1,x_2,x_3)=1/\{1+\exp(-(\beta_1x_1+\beta_2x_2+\beta_3x_3+\xi x_2x_3))\}, 
$$
where $\beta_j\sim \mcN(1,0.1^2)$. 
We consider the generalized linear model with logit link as the interested learning procedure, and the splitting ratio is fixed at $85\%$ for our proposed test. The sample size $n$ varies from $200$ to $4000$. 
The other settings keep the same as Simulation~1. 

Results are calculated based on $500$ replications.
The empirical performance with independent predictors is presented in Figures~\ref{figure_1}--\ref{figure_3}, and similar results for correlated predictors are reported in Appendix~\ref{app_sec_simu}. 
Overall, the SPARK maintains good size control across a wide range of scenarios, while HCZ may exhibit an inflated Type-I error rate when the sample size is small ($n=100,200$). 
In addition, our method exhibits power comparable to that of HCZ only for the high-dimensional model in Simulation~1 with a large sample size, whereas it achieves the highest power in most of the other scenarios under $H_1$. 


\begin{figure}[htbp]\centering
\includegraphics[width=1\textwidth]{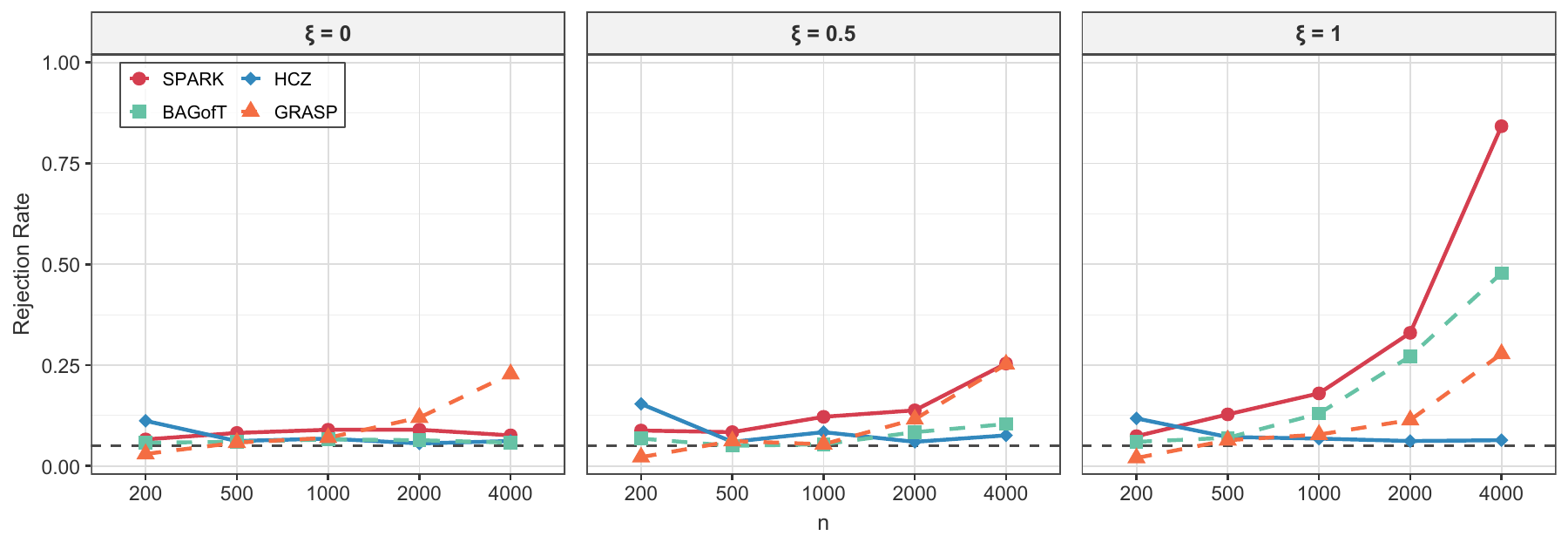}
\caption{Empirical Type-I error and power under Simulation~3 with independent predictors. The horizontal axis denotes the sample size $n$, and the vertical axis shows the rejection rate. The black dashed line indicates the nominal significance level of $0.05$. }\label{figure_3}
\end{figure}

\subsection{Performance on General Learning Procedures}\label{sec_52}
We evaluate the SPARK using different learning procedures in high-dimensional and low-dimensional, linear and nonlinear models as described in Simulations~4--6. 
We consider HCZ \citep{He2025AGA} as a competing method in these scenarios. 
We vary the data-splitting ratio as suggested in Section~\ref{sec_41}, while vary the ratio of HCZ as their suggestion. 
We use $\text{Ratio}_1,\text{Ratio}_2,\text{Ratio}_3$ to represent $20\%$, $25\%$, $50\%$ for our methods, and $50\%$, $75\%$, $90\%$ for theirs, respectively. 
We aggregate the result as described in Section~\ref{sec_42} using the Cauchy combination test \citep{Liu2020CauchyCT}. 
The sample size is fixed at $n = 1000$ and results are based on $100$ replications.

\textit{Simulation~4 (High-dimensional linear model).} 
Let $X=(x_1,\dots,x_p)^\top\sim\mcN(0,I_p)$ denote a $p$-dimensional predictor. The response variable $Y$ is generated as
$$
Y= \beta_1x_1+\beta_2x_2+\beta_3x_3+\beta_4x_4+\beta_5x_5+\epsilon, 
$$
where $\beta_j\sim\mcN(1,0.1^2)$ for $j=1,\dots,5$ and  $\epsilon\sim\mcN(0,1)$. We fix the dimension $p = 500$. 

\begin{figure}[htbp]\centering
\includegraphics[width=1\textwidth]{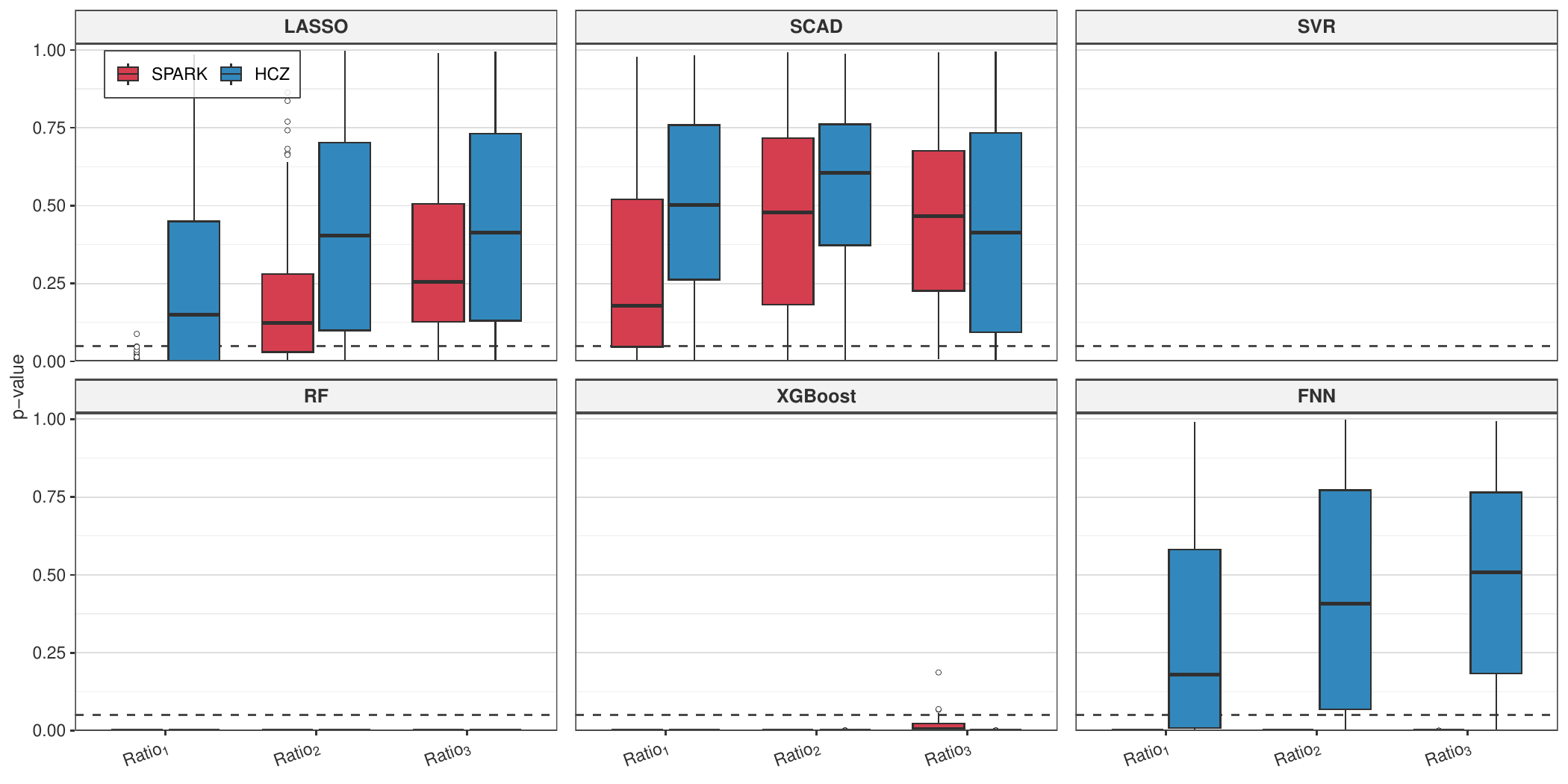}
\caption{Boxplot of the $p$-values in Simulation~4. The horizontal axis indicates different splitting ratios, while the vertical axis displays the corresponding $p$-values. 
The black dashed line marks the significance level at $0.05$.}\label{figure_4}
\end{figure}


\textit{Simulation~5 (High-dimensional nonlinear model).}
Let $X=(x_1,\dots,x_p)^\top\sim\mcN(0,\Sigma)$ denote a $p$-dimensional predictor with $\Sigma=(0.5^{|i-j|})_{p\times p}$. The response variable $Y$ is generated as
$$
Y= \beta_1x_1^2+\beta_2x_2+\beta_3x_3+\beta_4x_4+\beta_5x_5+\epsilon, 
$$
where $\beta_j\sim\mcN(1,0.1^2)$ for $j=1,\dots,5$ and  $\epsilon\sim\mcN(0,1)$. We fix the dimension $p = 500$. 

\begin{figure}[htbp]\centering
\includegraphics[width=1\textwidth]{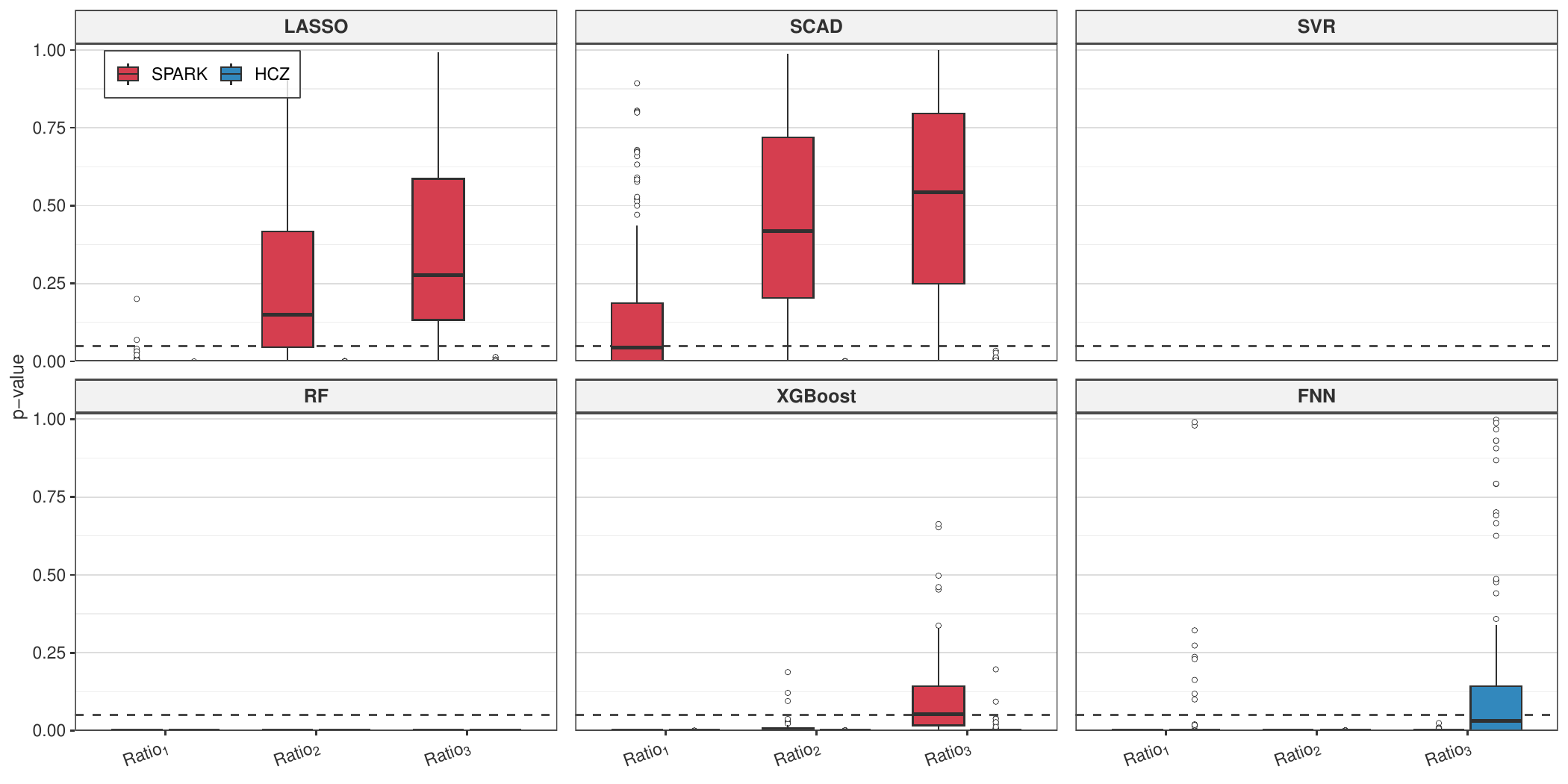}
\caption{Boxplot of the $p$-values in Simulation~5. The horizontal axis indicates different splitting ratios, while the vertical axis displays the corresponding $p$-values. 
The black dashed line marks the significance level at $0.05$. }\label{figure_5}
\end{figure}

\textit{Simulation~6 (Low-dimensional nonlinear model).} 
Let $X=(x_1,x_2)^\top\sim\mcN(0,\Sigma)$ denote a $2$-dimensional predictor with $\Sigma=(0.5^{|i-j|})_{2\times 2}$. The response variable $Y$ is generated as
$$
Y= \beta_1x_1^2+\beta_2\exp(x_2)+\epsilon, 
$$
where $\beta_j\sim\mcN(1,0.1^2)$ for $j=1,2$ and $\epsilon\sim\mcN(0,1)$. 

\begin{figure}[htbp]\centering
\includegraphics[width=1\textwidth]{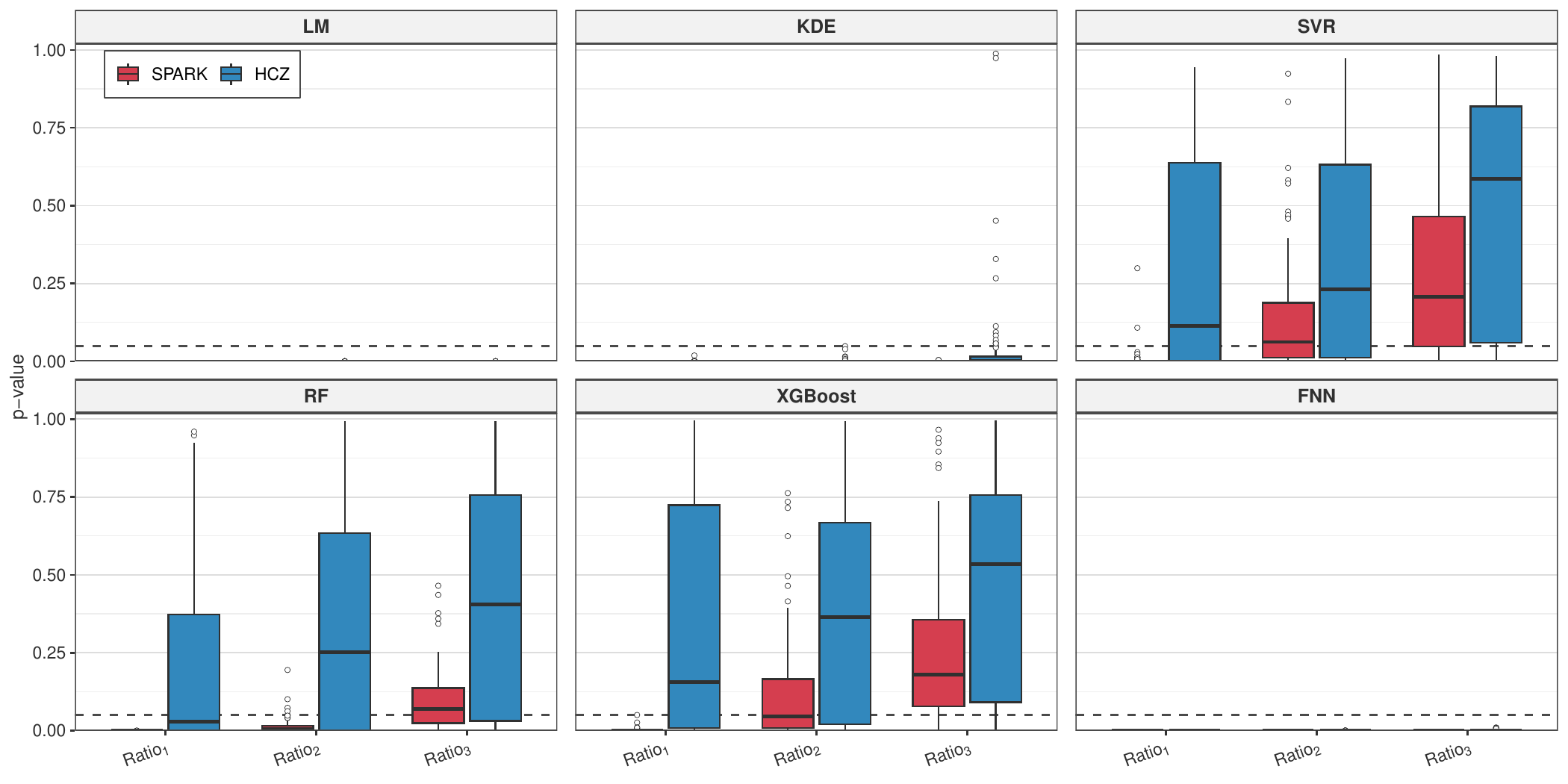}
\caption{Boxplot of the $p$-values in Simulation~6. The horizontal axis indicates different splitting ratios, while the vertical axis displays the corresponding $p$-values. 
The black dashed line marks the significance level at $0.05$. }\label{figure_6}
\end{figure}

Following \cite{He2025AGA}, we apply a range of learning procedures to the data generated from Simulations~4--6, including random forest (RF) with $100$ trees \citep{Breiman2001RandomF}, XGBoost with default settings from the \texttt{xgboost} R package \citep{Chen2016XGBoostAS}, a feedforward neural network (FNN) with one hidden layer containing $80$ neurons \citep{Schmidhuber2015DeepLI}, support vector regression (SVR) with a radial basis function kernel \citep{Smola2004ATO}. 
We also consider LASSO \citep{Tibshirani1996RegressionSA} and SCAD \citep{Fan2001VariableSV} for high-dimensional settings in Simulations~4--5, and ordinary least squares (LM) and kernel density estimation (KDE) for the low-dimensional setting in Simulation~6. 

Figures~\ref{figure_4}--\ref{figure_6} show the boxplots of the $p$-values in Simulations~4--6. 
The SPARK suggests that, LASSO and SCAD exhibit superior performance in high-dimensional settings in Simulations~4--5, with SCAD performs better than LASSO. 
In Simulation~6, the SPARK indicates that the estimations using SVR and XGBoost converge to the true conditional mean at a slow rate, with SVR performs slightly better, while the other learning procedures fail to converge to the true one. 

Moreover, our results differ from those of HCZ in several cases. For example, their results suggest that FNN achieves a fast convergence rate in Simulation~4, whereas the SPARK implies that FNN fails to converge to true conditional mean. 
To ensure a fair comparison, we split the dataset evenly, using one half for training and the other half for testing. The resulting training mean squared error (MSEs) and test MSEs are reported in Figure~\ref{figure_mse}, which shows that the test MSEs of FNN are substantially higher than those of LASSO and SCAD. 
In this sense, the SPARK is consistent with the conclusions drawn from cross-validated MSEs. Similar patterns are observed in Simulations~5--6 in Figure~\ref{figure_mse}, where the test MSE aligns more closely with our method. 

\begin{figure}[htbp]\centering
\includegraphics[width=1\textwidth]{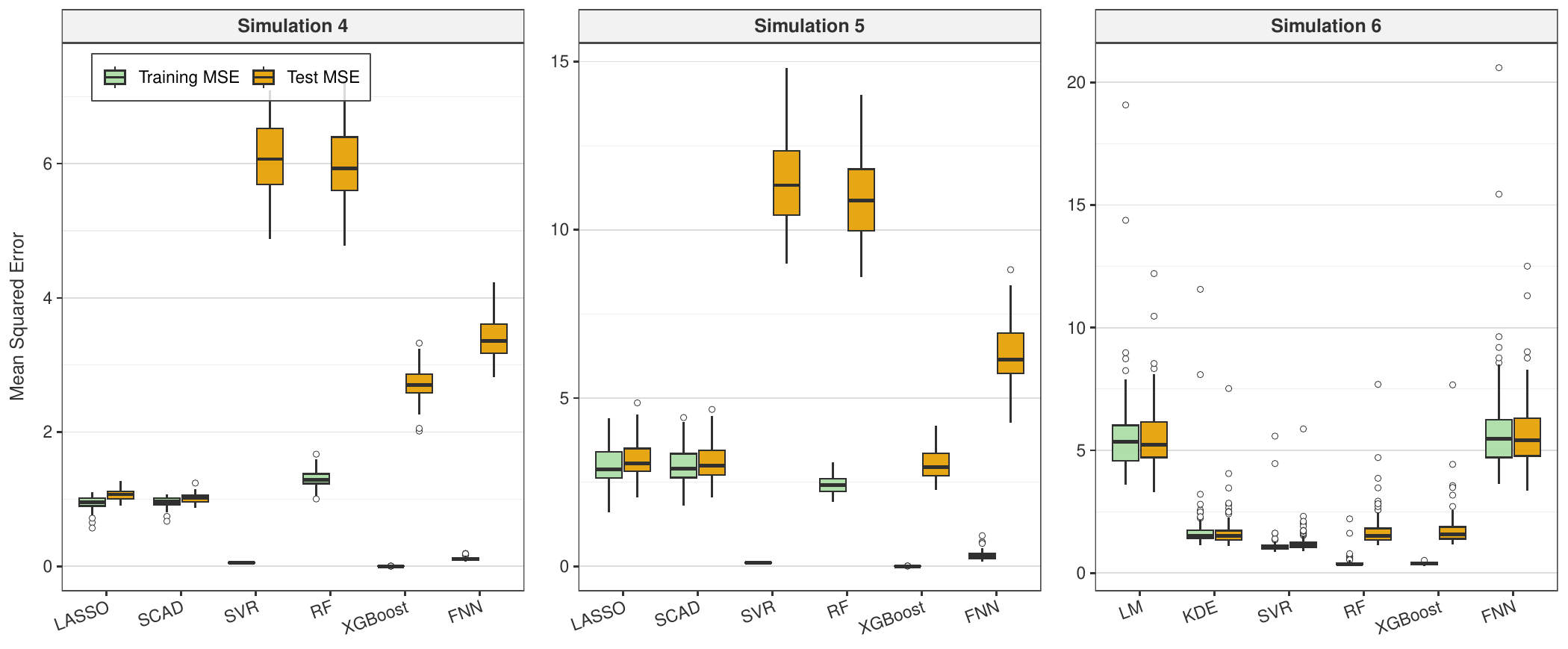}
\caption{Boxplots of the training MSEs and test MSEs from Simulations~4--6. The horizontal axis indicates different learning procedure. Results are based on $100$ replications.}\label{figure_mse}
\end{figure}

\section{Real Data Analysis}\label{sec_06}
We demonstrate the practical applicability of our method through three real-world datasets. 
We study a regression problem with a large sample size and low-dimensional predictors in Section~\ref{sec_61}, and a classification task with a large sample size and high-dimensional predictors in Section~\ref{sec_62}. Appendix~\ref{app_sec_simu_2} includes analysis about a regression problem with a small sample size and high-dimensional predictors. 

\subsection{Wine Quality Dataset}\label{sec_61}
The wine quality dataset, originally introduced by \cite{CORTEZ2009547}, has been widely studied using a range of methods, including multiple regression, neural networks, and support vector regression. The dataset is publicly available at \url{https://archive.ics.uci.edu/dataset/186/wine+quality}.

In our analysis, the response variable is the quality rating of white wine. The predictors include fixed acidity, volatile acidity, citric acid, residual sugar, chlorides, free sulfur dioxide, density, and pH. The dataset contains $n=4898$ observations. 
We assess the performance  LASSO, SCAD, SVR, RF, XGBoost, and FNN under different data-splitting ratios, following the setup described in Section~\ref{sec_52}. 
The test results are presented in Table~\ref{tab_wine}, where the $p$-values are obtained based on $2$ multiple-splitting procedures, while the MSEs are averaged over $100$ replications.

\newcommand{\pval}[3]{\makecell[c]{#1\\{\scriptsize(#2,#3)}}}
{\renewcommand{\baselinestretch}{1.1}
\begin{table}[htbp]
\centering
\caption{Test results and MSEs for the wine quality dataset. For $p$-values, each entry is reported as median (lower $25\%$-quantile, upper $25\%$-quantile).}\label{tab_wine}
\setlength{\tabcolsep}{3.5pt}
\renewcommand{\arraystretch}{1.0}
\small
\begin{adjustbox}{max width=\textwidth}
\begin{tabular}{lccccccc}
\toprule
 & Ratio & LASSO & SCAD & SVR & RF & XGBoost & FNN \\
\midrule
\multirow{4}{*}{SPARK}
& 25\% & \pval{.000}{.000}{.000} & \pval{.000}{.000}{.000} & \pval{.010}{.000}{.044} & \pval{.000}{.000}{.000} & \pval{.008}{.000}{.049} & \pval{.000}{.000}{.000} \\
& 50\% & \pval{.000}{.000}{.000} & \pval{.000}{.000}{.000} & \pval{.079}{.027}{.173} & \pval{.000}{.000}{.006} & \pval{.111}{.036}{.220} & \pval{.000}{.000}{.000} \\
& 75\% & \pval{.000}{.000}{.000} & \pval{.000}{.000}{.000} & \pval{.220}{.096}{.481} & \pval{.040}{.014}{.094} & \pval{.253}{.063}{.491} & \pval{.000}{.000}{.000} \\
\midrule
\multirow{3}{*}{HCZ}
& 50\% & \pval{.000}{.000}{.000} & \pval{.000}{.000}{.000} & \pval{.008}{.000}{.172} & \pval{.000}{.000}{.000} & \pval{.016}{.000}{.241} & \pval{.000}{.000}{.000} \\
& 75\% & \pval{.000}{.000}{.000} & \pval{.000}{.000}{.000} & \pval{.483}{.089}{.714} & \pval{.000}{.000}{.000} & \pval{.460}{.119}{.791} & \pval{.000}{.000}{.000} \\
& 90\% & \pval{.006}{.000}{.075} & \pval{.023}{.001}{.147} & \pval{.512}{.249}{.732} & \pval{.052}{.001}{.380} & \pval{.557}{.354}{.731} & \pval{.000}{.000}{.000} \\
\midrule
\multirow{3}{*}{MSE (Train)}
& 25\% & .557 & .557 & .212 & .118 & .103 & .747 \\
& 50\% & .561 & .561 & .250 & .106 & .175 & .685 \\
& 75\% & .563 & .562 & .267 & .097 & .216 & .675 \\
\midrule
\multirow{3}{*}{MSE (Test)}
& 25\% & .575 & .575 & .581 & .473 & .521 & .751 \\
& 50\% & .572 & .571 & .518 & .424 & .473 & .692 \\
& 75\% & .572 & .570 & .482 & .387 & .447 & .682 \\
\bottomrule
\end{tabular}
\end{adjustbox}
\end{table}
}

It implies that only SVR and XGBoost converge to the real conditional mean at a relatively slow rate, while the other four learning procedures may not fit the data well. 
The results are consistent with those of HCZ and show higher power. 
Since HCZ is restricted to high-dimensional predictors and Simulation~2 further demonstrates superior empirical performance of our methods, our conclusions are more convincing and reliable.
This conclusion is further supported by \cite{CORTEZ2009547}, where SVR achieves the best performance in terms of mean absolute deviation and accuracy. Moreover, the test MSEs of SVR and XGBoost are significantly lower than those of LASSO, SCAD, and FNN.

It is also worth noting that RF achieves the lowest test MSEs across several splitting ratios. 
This mismatch may be explained by that RF relies on local averaging, which enables it to capture structure well in high-density regions and achieve a low test MSE even when the bias is relatively large in boundary or sparse regions. 
It also suggests that our assessment and the MSE reflect different aspects of model performance.

\subsection{MNIST Dataset}\label{sec_62}
In this subsection, we study a classification task using the MNIST dataset, a dataset contains a vast collection of handwritten digits. 
Following \cite{He2025AGA}, we focus on two visually similar digits, ``4'' and ``9'', and randomly select a subset of $n = 1000$ images with  $28\times28$ pixels. 

We employed several classification algorithms, including XGBoost and five FNN configurations, which are designed to differ in network depth, hidden-layer dimensionality, and activation function. Specifically, FNN-1 and FNN-2 each consist of a single hidden layer with one neuron, using ReLU and sigmoid activation functions, respectively. FNN-3 comprises two hidden layers with $64$ and $16$ neurons and used ReLU activation. FNN-4 and FNN-5 adopt a three-hidden-layer architecture with $128$, $64$, and $16$ neurons, using ReLU and sigmoid activation functions, respectively.


Table~\ref{tab_mnist} reports the performance of the competing methods under different splitting ratios. Compared with the existing goodness-of-fit procedures, the SPARK generally provides stronger evidence against the null hypothesis and exhibits greater sensitivity to model misspecification. In particular, it remains capable of detecting lack of fit in settings where HCZ and GRASP yield substantially weaker rejection signals. 
Moreover, the SPARK suggests that only the deeper neural networks, namely FNN-3, FNN-4, and FNN-5, exhibit a slow convergence rate, whereas the remaining methods fail to converge. 
This finding is consistent with the test MSE results, since these three neural networks achieve the three lowest test MSEs across the splitting ratios.

Another analysis for a cortisol stress-reactivity dataset, presented in Appendix~\ref{app_sec_simu_2}, examines our method for data with a small sample size and high-dimensional predictors. The results lead to a similar conclusion and further confirm the reliability of our method.

{\renewcommand{\baselinestretch}{1.1}
\begin{table}[htbp]
\centering
\caption{Test results and MSEs for the MNIST dataset. For $p$-values, each entry is reported as median (lower $25\%$-quantile, upper $25\%$-quantile).}
\label{tab_mnist}
\setlength{\tabcolsep}{3.5pt}
\renewcommand{\arraystretch}{1.0}
\begin{adjustbox}{max width=\textwidth}
\small
\begin{tabular}{lccccccc}
\toprule
 & Ratio & XGBoost & FNN-1 & FNN-2 & FNN-3 & FNN-4 & FNN-5 \\
\midrule
\multirow{3}{*}{SPARK} 
 & 25\% & \pval{.000}{.000}{.000} & \pval{.000}{.000}{.000} & \pval{.000}{.000}{.000} & \pval{.000}{.000}{.000} & \pval{.000}{.000}{.004} & \pval{.000}{.000}{.000} \\
 & 50\% & \pval{.000}{.000}{.004} & \pval{.000}{.000}{.000} & \pval{.000}{.000}{.000} & \pval{.013}{.002}{.077} & \pval{.013}{.000}{.081} & \pval{.036}{.008}{.114} \\
 & 75\% & \pval{.057}{.010}{.177} & \pval{.000}{.000}{.000} & \pval{.000}{.000}{.000} & \pval{.125}{.026}{.424} & \pval{.118}{.047}{.336} & \pval{.230}{.106}{.458} \\
\midrule
\multirow{3}{*}{HCZ} & 50\% & \pval{.295}{.004}{.677} & \pval{.000}{.000}{.000} & \pval{.000}{.000}{.000} & \pval{.417}{.160}{.644} & \pval{.435}{.213}{.740} & \pval{.612}{.200}{.815} \\
 & 75\% & \pval{.548}{.372}{.776} & \pval{.000}{.000}{.000} & \pval{.000}{.000}{.000} & \pval{.523}{.256}{.681} & \pval{.361}{.140}{.581} & \pval{.539}{.333}{.762} \\
 & 90\% & \pval{.553}{.335}{.740} & \pval{.000}{.000}{.000} & \pval{.000}{.000}{.000} & \pval{.503}{.247}{.732} & \pval{.358}{.180}{.709} & \pval{.479}{.294}{.699} \\
\midrule
\multirow{3}{*}{BAGofT} & 50\% & \pval{.050}{.003}{.362} & \pval{.000}{.000}{.000} & \pval{.000}{.000}{.000} & \pval{.000}{.000}{.000} & \pval{.000}{.000}{.000} & \pval{.001}{.000}{.025} \\
 & 75\% & \pval{.196}{.020}{.559} & \pval{.000}{.000}{.000} & \pval{.000}{.000}{.000} & \pval{.000}{.000}{.000} & \pval{.000}{.000}{.003} & \pval{.009}{.001}{.069} \\
 & 90\% & \pval{.410}{.155}{.705} & \pval{.000}{.000}{.000} & \pval{.000}{.000}{.000} & \pval{.000}{.000}{.024} & \pval{.000}{.000}{.058} & \pval{.085}{.007}{.318} \\
\midrule
\multirow{3}{*}{GRASP} & 50\% & \pval{.016}{.004}{.041} & \pval{.008}{.001}{.035} & \pval{.000}{.000}{.000} & \pval{.128}{.023}{.340} & \pval{.088}{.032}{.275} & \pval{.036}{.007}{.117} \\
 & 75\% & \pval{.140}{.031}{.274} & \pval{.042}{.006}{.198} & \pval{.000}{.000}{.008} & \pval{.211}{.082}{.539} & \pval{.236}{.064}{.479} & \pval{.169}{.070}{.343} \\
 & 90\% & \pval{.351}{.130}{.554} & \pval{.249}{.099}{.449} & \pval{.037}{.010}{.209} & \pval{.375}{.231}{.679} & \pval{.424}{.181}{.720} & \pval{.335}{.187}{.598} \\
\midrule
\multirow{3}{*}{MSE (Train)} & 25\% & .000 & .227 & .214 & .007 & .012 & .018 \\
 & 50\% & .000 & .221 & .184 & .010 & .008 & .014 \\
 & 75\% & .000 & .200 & .154 & .008 & .009 & .016 \\
\midrule
\multirow{3}{*}{MSE (Test)} & 25\% & .075 & .236 & .220 & .067 & .065 & .047 \\
 & 50\% & .055 & .226 & .190 & .051 & .045 & .040 \\
 & 75\% & .047 & .205 & .159 & .041 & .037 & .038 \\
\bottomrule
\end{tabular}
\end{adjustbox}
\end{table}
}

\section{Conclusion}\label{sec_07}
This paper introduces a novel test procedure to assess the goodness-of-fit of general learning procedures. The testing procedure splits the dataset, with a subset used for training an interested learning procedure and the other is used for evaluating the goodness-of-fit. 
We project the residuals on the space of the predictors to evaluate whether any information left. Our proposal characterizes all the projection directions, whereas \cite{He2025AGA} selected several special projection directions. Future avenues include taking the predictors and the responses as complex objects. 



\newpage
\begin{appendices}
\label{app}
\allowdisplaybreaks[3]

\renewcommand{\theequation}{\thesection.\arabic{equation}}

\setcounter{table}{0}
\renewcommand{\thetable}{\thesection\arabic{table}}
\setcounter{figure}{0}
\renewcommand{\thefigure}{\thesection\arabic{figure}}
\setcounter{algorithm}{0}
\renewcommand{\thealgorithm}{\thesection\arabic{algorithm}}
\setcounter{assumption}{0}
\renewcommand{\theassumption}{\thesection\arabic{assumption}}
\section{Overview}
This supplementary material includes technical details and proofs omitted in the main text, and additional simulations and data analysis. The remaining material is organized as follows.
\begin{itemize}
    \item In Section~\ref{app_sec_02}, we present technical proofs for the results in the main text, including the reformulations of the hypothesis testing (Section~\ref{app_sec_021} and Section~\ref{app_sec_lemma_02}), limiting distributions for a general $f$ (Section~\ref{app_sec_022}), limiting distributions for the kernel-based projection (Section~\ref{app_sec_023}), asymptotic properties of the wild bootstrap under the null (Section~\ref{app_sec_024}), and the alternatives (Section~\ref{app_sec_025}). 
    \item Section~\ref{app_sec_04} further discusses the asymptotic behaviour of the statistic $T_f$ under the alternatives.
    \item Section~\ref{app_sec_E} includes the dimension-agnostic theories for the kernel-based statistic. 
    \item Section~\ref{app_sec_05} contains additional numerical results, including the rejection rate of correlated predictors (Section~\ref{app_sec_simu}), and real data analysis for the cortisol stress reactivity dataset (Section~\ref{app_sec_simu_2}). 
\end{itemize}

\section{Technical Proofs}\label{app_sec_02}
\paragraph{Notation.}
For a probability measure $\mu$ on $\mathbb{R}^d$ and $1 \le p \le \infty$, let $L_p(\mu) = \{ g : \mathbb{R}^d \to \mathbb{R} \mid \|g\|_{L_p(\mu)} < \infty \}$, where $\|g\|_{L_p(\mu)}$ denotes the usual $L_p$-norm. In particular, $L_\infty(\mu)$ denotes the space of essentially bounded functions, and $\|g\|_{L_\infty(\mu)}$ denotes the essential supremum with respect to $\mu$. 
For two sequences $a_n$ and $b_n$, the notation $a_n = O(b_n)$ or $a_n \lesssim b_n$ indicates that there exists a constant $C>0$ such that $|a_n| \leq C|b_n|$ for all sufficiently large $n$. The notation $a_n \asymp b_n$ means that $a_n = O(b_n)$ and $b_n = O(a_n)$. 
Furthermore, $a_n = o(b_n)$ denotes that $a_n/b_n \to 0$ as $n \to \infty$. 
For a random sequence $A_n$, the notation $A_n = O_p(b_n)$ means that $A_n/b_n$ is bounded in probability, and $A_n = o_p(b_n)$ means that $A_n/b_n \to 0$ in probability. Similarly, $A_n = \omega_p(b_n)$ and $A_n = \Omega_p(b_n)$ signify that $b_n = o_p(A_n)$ and $b_n = O_p(A_n)$, respectively.
Throughout the proof, $C_1,C_2,\dots$ stand for some positive constants. 

\subsection{Proof of Lemma~\ref{lemma_01}}\label{app_sec_021}
\begin{proof}
Let $\sigma(\mcD_n)$ denote the $\sigma$-field generated by $\mcD_n$ and $\mcD_n(\omega)$ denote the realization of $\mcD_n$ at the sample point $\omega$. 
Denote $g_{n,\omega}(x)=m_{\mcD_n(\omega)}(x)-m(x)$ for any fixed sample point $\omega$, and we write $g_{n}(x)=m_{\mcD_n}(x)-m(x)$ as the corresponding random variable in $\sigma(\mcD_n)$. We first show that 
\begin{equation}\label{app_eq_01}
    \sup_{f_{\mcD_n}\in\mcF_n}|E\{\varepsilon_n f_{\mcD_n}(X)|\mcD_n\}(\omega)|=\|g_{n,\omega}\|_{L_\infty(P_X)}\ a.e.,
\end{equation}
where $\|g_{n,\omega}\|_{L_\infty(P_X)}:=\inf\{a\geq0:P(|g_{n,\omega}(X)|>a)=0\}$ is the essential supremum of $g_{n,\omega}$ with respect to measure $P_X$. In fact, note that 
$$
    \varepsilon_n= \epsilon + m(X)-m_{\mcD_n}(X)= \epsilon-g_n(X). 
$$
For $f_{\mcD_n}\in\mcF_n$, $E\{|\epsilon f_{\mcD_n}(X)|\,|\mcD_n\}<\infty$ almost surely, and thus
$$
\int |g_{n,\omega}(x)f_{\mcD_n(\omega)}(x)|\,\dd P_X(x)
    \le
   \|g_{n,\omega}\|_{L_\infty(P_X)}\int |f_{\mcD_n(\omega)}(x)|\,\dd P_X(x)
    \le \|g_{n,\omega}\|_{L_\infty(P_X)}<\infty.
$$
It follows that $E\{|\varepsilon_n f_{\mcD_n}(X)|\,|\mcD_n\}<\infty$ a.e. 
Therefore, $E\{\varepsilon_n f_{\mcD_n}(X)|\mcD_n\}$ is well-defined. 
Moreover, note that since 
$$
    E(\varepsilon_n|X,\mcD_n)= E(\epsilon|X,\mcD_n) + m(X) - m_{\mcD_n}(X)=m(X) - m_{\mcD_n}(X)=-g_n(X)\ a.e,  
$$
it follows that 
$$
E(\varepsilon_nf_{\mcD_n}(X)|X,\mcD_n)=-g_n(X)f_{\mcD_n}(X)\ a.e.
$$
Taking conditional expectation with respect to $\mcD_n$ gives
$$
E\{\varepsilon_n f_{\mcD_n}(X)|\mcD_n\}=E[E\{\varepsilon_n f(X)|X,\mcD_n\}|\mcD_n]=-E\{g_n(X) f_{\mcD_n}(X)|\mcD_n\}\ a.e.
$$
Since $X$ is independent of $\mcD_n$, it holds that 
$$
E\{g_n(X) f_{\mcD_n}(X)|\mcD_n\}(\omega)=\int g_{n,\omega}(x)f_{\mcD_n(\omega)}(x)\dd P_X(x)
$$
for almost all $\omega$.
Taking absolute values and then the supremum over all measurable $f_{\mcD_n}$ with $f_{\mcD_n}\in\mcF_n$, we have 
\begin{equation}\label{app_eq_02}
    \sup_{f_{\mcD_n}\in\mcF_n}|E\{g_n(X) f_{\mcD_n}(X)|\mcD_n\}(\omega)|=\sup_{f_{\mcD_n}\in\mcF_n}\bigg|\int g_{n,\omega}(x)f_{\mcD_n(\omega)}(x)\dd P_X(x)\bigg|. 
\end{equation}

We now show that the right-hand side of \eqref{app_eq_02} equals $\|g_{n,\omega}\|_{L_\infty(P_X)}$ for almost all $\omega$. For the upper bound of the right-hand side, we have
$$
\begin{aligned}
    \bigg|\int g_{n,\omega}(x)f_{\mcD_n(\omega)}(x)\dd P_X(x)\bigg|&\leq\int |g_{n,\omega}(x)||f_{\mcD_n(\omega)}(x)|\dd P_X\\
    &\leq \|g_{n,\omega}\|_{L_\infty(P_X)}\int |f_{\mcD_n(\omega)}(x)|\dd P_X\\
    &\leq \|g_{n,\omega}\|_{L_\infty(P_X)}, 
\end{aligned}
$$
where the last inequality holds if $\|f_{\mcD_n(\omega)}\|_{L_1(P_X)}\leq1$. 

For the other direction, we fix a sample point $\omega$ outside a null set such that $\|g_{n,\omega}\|_{L_\infty(P_X)}<\infty$. 
For any positive rational number $a<\|g_{n,\omega}\|_{L_\infty(P_X)}$, by the definition of essential supremum, $A_a(\omega):=\{x:|g_{n,\omega}(x)|>a\}$ has a positive measure, i.e., $p_a(\omega):=P_X(A_a(\omega))>0$. 
Define 
$$
f^{a}_{\mcD_n(\omega)}(x):=\frac{\operatorname{sgn}(g_{n,\omega}(x))\mathbbm{1}_{A_a(\omega)}(x)}{p_a(\omega)},$$
where $\operatorname{sgn}(\cdot)$ denotes the sign function and $\mathbbm{1}_{A_a(\omega)}(\cdot)$ is the indicator function of event $A_a(\omega)$ at $x$. 
For any given $\omega$, $f^{a}_{\mcD_n(\omega)}$ is $\mcB(\mbS)$-measurable, where $\mcB(\mbS):=\{B\cap\mbS:B\in\mcB(\mbR^p)\}$ and $\mcB(\mbR^p)$ denotes the Borel $\sigma$-field on $\mbR^p$. 
A calculation yields that $\|f^{a}_{\mcD_n(\omega)}\|_{L_1(P_X)}=1$.
Thus $f^{a}_{\mcD_n(\omega)}\in\mcF_n$. 
Moreover, for every $a<\|g_{n,\omega}\|_{L_\infty(P_X)}$, it holds that 
$$
\bigg|\int g_{n,\omega}(x)f^{a}_{\mcD_n(\omega)}(x)\dd P_X(x)\bigg|=\frac{1}{p_a(\omega)}\int_{A_a(\omega)} |g_{n,\omega}(x)|\dd P_X(x)\geq a. 
$$
Then, taking $a\uparrow \|g_{n,\omega}\|_{L_\infty(P_X)}$, we have 
$$
    \sup_{f_{\mcD_n}\in\mcF_n}\bigg|\int g_{n,\omega}(x)f_{\mcD_n(\omega)}(x)\dd P_X(x)\bigg|\geq \|g_{n,\omega}\|_{L_\infty(P_X)}\ a.e. 
$$
Combining the two bounds together and \eqref{app_eq_02} together yields \eqref{app_eq_01}.


To complete the proof, it is sufficient to show that the null hypothesis is equivalent to 
$$
\|g_{n,\omega}\|_{L_\infty(P_X)}=O_p(r_n).
$$
First we suppose that there exists a set $\mbM_n^0(\omega)\subseteq\mbS$ with $P(X\in\mbM_n^0(\omega))=1$ such that 
$$\sup_{x\in\mbM_n^0(\omega)}|g_{n,\omega}(x)|=O_p(r_n).$$ 
Since $\mbM_n^0(\omega)$ has full $P_X$-measure, by the definition of essential supremum, it holds that 
$$
\|g_{n,\omega}\|_{L_\infty(P_X)}\leq\sup_{x\in\mbM_n^0}|g_{n,\omega}(x)|.
$$
Hence $\|g_n\|_{L_\infty(P_X)}=O_p(r_n)$. 

For the other direction, suppose that $\|g_{n,\omega}\|_{L_\infty(P_X)}=O_p(r_n)$. Define 
$$\mbM_n^0(\omega):=\{x\in\mbS:|g_{n,\omega}(x)|\leq\|g_{n,\omega}\|_{L_\infty(P_X)}\}.$$ 
Then the definition of essential supremum, $P(X\in\mbM_n^0(\omega))=1$ for almost all $\omega$. It follows that 
$$
\sup_{x\in\mbM_n^0(\omega)}|g_{n,\omega}(x)|\leq\|g_{n,\omega}\|_{L_\infty(P_X)}=O_p(r_n),\ a.e.,
$$
which yields the first statement of Lemma~\ref{lemma_01} about the equivalence. 

For the second statement on the necessary condition, by Jensen's inequality and the triangle inequality, it holds that 
$$
\|f_{\mcD_n}(X)-E(f_{\mcD_n}(X)|\mcD_n)\|_{L_1(P_X)}\leq2\|f_{\mcD_n}(X)\|_{L_1(P_X)}, 
$$ 
for $\mcD_n$ a.e., which completes the proof. 
\end{proof}

\subsection{Proof of Theorem~\ref{th_01}}\label{app_sec_022}
\begin{proof}
It is noted that $T_f$ can be decomposed as 
\begin{align*}
    T_f&=\frac{1}{n_2}\sum_{i=1}^{n_2}\big\{Y_{2,i}-m(X_{2,i})+m(X_{2,i})-m_{\mcD_{n_1}}(X_{2,i})\big\}\big\{f(X_{2,i})-E(f(X))+E(f(X))-\hat m_f\big\}\\
    &=\frac{1}{n_2}\sum_{i=1}^{n_2}\big\{Y_{2,i}-m(X_{2,i})\big\}\big\{f(X_{2,i})-E(f(X))\big\}+
    \frac{1}{n_2}\sum_{i=1}^{n_2}\big\{Y_{2,i}-m(X_{2,i})\big\}\big\{E(f(X))-\hat m_f\big\}
    \\
    &\quad\ 
    +\frac{1}{n_2}\sum_{i=1}^{n_2}\big\{m(X_{2,i})-m_{\mcD_{n_1}}(X_{2,i})\big\}\big\{f(X_{2,i})-E(f(X))\big\}\\
    &\quad\ +\frac{1}{n_2}\sum_{i=1}^{n_2}\big\{m(X_{2,i})-m_{\mcD_{n_1}}(X_{2,i})\big\}\big\{E(f(X))-\hat m_f\big\}\\
    &=:A_1+A_2+A_3+A_4. 
\end{align*}
We now calculate the four terms in order. Note that since 
\begin{align*}
    E(A_1)&=E\big[\{Y-m(X)\}\{f(X)-E(f(X))\}\big]\\
    &=E\big[E[\{Y-m(X)\}\{f(X)-E(f(X))\}|X]\big]\\
    &=E\big[E[\{Y-m(X)\}|X]\{f(X)-E(f(X))\}\big]=0, 
\end{align*}
a direct use of central limit theorem yields that, as $n_2\to\infty$, 
\begin{equation}\label{app_eq_03}
    {\sqrt{n_2}A_1}=\frac{1}{\sqrt{n_2}}\sum_{i=1}^{n_2}\{Y_{2,i}-m(X_{2,i})\}\{f(X_{2,i})-E(f(X))\}\overset{d}{\to}\mcN(0,V_f), 
\end{equation}
where $V_f=E[\{Y-m(X)\}^2\{f(X)-E(f(X))\}^2]$. 
Similarly, it is noted that 
\begin{align*}
    E(A_2)=E\big[\{Y-m(X)\}\{E(f(X))-\hat m_f\}\big]=E\{Y-m(X)\}E\{E(f(X))-\hat m_f\}=0,
\end{align*}
where the second equality holds since the two terms are independent. 
Then, a calculation leads to 
\begin{align*}
    E(A_2^2)&=E\bigg[\frac{1}{n_2^2}\sum_{i=1}^{n_2}\big\{Y_{2,i}-m(X_{2,i})\big\}^2\big\{E(f(X))-\hat m_f\big\}^2\bigg]\\
    &={n_2^{-1}}E\{Y-m(X)\}^2E\big\{E(f(X))-\hat m_f\big\}^2\\
    &=O(n_1^{-1}n_2^{-1}).
\end{align*}
By Markov's inequality, it follows that 
\begin{equation}\label{app_eq_05}
A_2=O_p\big(n_1^{-1/2}n_2^{-1/2}\big). 
\end{equation}

Under the null hypothesis, H\"older's inequality yields that 
$$
E|A_3|\leq \sup_{x\in\mbM_n^0}|m_{\mcD_{n_1}}(x)-m(x)|\cdot E|f(X_{2,i})-E(f(X))|\ a.e., 
$$
which is followed by 
\begin{equation}\label{app_eq_04}
A_3=O_p(r_{n_1}). 
\end{equation}
Moreover, under the null hypothesis, we have 
$$
|A_4|\leq \sup_{x\in\mbM_n^0}|m_{\mcD_{n_1}}(x)-m(x)|\cdot|E(f(X))-\hat m_f|\ a.e., 
$$
which is followed by 
\begin{equation}\label{app_eq_06}
A_4=O_p\big(r_{n_1}n_1^{-1/2}\big). 
\end{equation}

Combining the \eqref{app_eq_03}--\eqref{app_eq_06} together, we have 
$$
\sqrt{n_2}T_f=\sqrt{n_2}A_1+O_p\big(n_1^{-1}\big)+O_p\big(r_{n_1}n_2^{1/2}\big)+O_p\big(r_{n_1}n_1^{-1/2}n_2^{1/2}\big)\overset{d}{\to}\mcN(0,V_f). 
$$
The last convergence uses the condition that $n_2=o(r_{n_1}^{-2})$. 
The proof is then completed. 
\end{proof}

\subsection{Proof of Lemma~\ref{lemma_02}}\label{app_sec_lemma_02}
\begin{proof}
    We prove the equivalence between 
\begin{equation}\label{app_eq_10}
    \sup_{f_{\mcD_n}\in\mcH_k^0}E\big[\varepsilon_n\{f_{\mcD_n}(X)-E(f_{\mcD_n}(X)|\mcD_n)\}|\mcD_{n}\big]=O_p(r_n),
\end{equation}
with 
\begin{equation}\label{app_eq_11}
\sup_{f_{\mcD_n}\in\mcF_n}E\big[\varepsilon_n\{f_{\mcD_n}(X)-E(f_{\mcD_n}(X)|\mcD_n)\}|\mcD_{n}\big]=O_p(r_n),
\end{equation}
with some appropriate kernel $k$. The proof includes two steps, which are formulated as Lemma~\ref{app_le_01} and \ref{app_le_02}, respectively. 
For the first step, by Lemma~\ref{app_le_01}, we illustrate that the metric $\|\cdot\|_{\mcH_k}$ is equivalent to $\|\cdot\|_{L_1(P_X)}$. It follows that \eqref{app_eq_10} is equivalent to 
\begin{equation}\label{app_eq_12}
\sup_{\substack{f_{\mcD_n(\omega)}\in\mcH_k\\\|f_{\mcD_n(\omega)}\|_{L_1(P_X)}\leq1}}E\big[\varepsilon_n\{f_{\mcD_n}(X)-E(f_{\mcD_n}(X)|\mcD_n)\}|\mcD_{n}\big]=O_p(r_n),
\end{equation}
for every $\omega$. 
In the second step, we proof that the equivalence holds if $\mcH_k$ is dense in $L_1(P_X)$ and the metric $\|\cdot\|_{\mcH_k}$.
Recall that 
    $$
    \mcF_n=\{f_{\mcD_n}\mid f_{\mcD_n}:\mbS\to \mbR, \|f_{\mcD_n}\|_{L_1(P_X)}\leq1\text{ for every given }\mcD_n\}.
    $$
    We further denote 
    $$
    \mcF_n(\omega):= \{f_{\mcD_n(\omega)}\mid f_{\mcD_n}\in\mcF_n\}, 
    $$
    for any fixed sample point $\omega$. Since $\mcF_n(\omega)$ is a subset of $L_1(P_X)$, it follows that $\mcH_k\cap \mcF_n(\omega)$ is dense in $\mcF_n(\omega)$ with respect to $L_1$-norm.
By Lemma~\ref{app_le_02}, \eqref{app_eq_12} is equivalent to \eqref{app_eq_11}, which completes the proof. 
\end{proof}

\begin{lemma}\label{app_le_01}
    Suppose $\mbS$ is a compact subset of $\mbR^p$. Suppose $k:\mbS\times \mbS\to\mbR$ is a continuous symmetric positive-definite kernel function with $k(x,x)>0$ for any $x\in\mbS$. Then $\|\cdot\|_{L_1(P_X)}$ and $\|\cdot\|_{\mcH_k}$ are uniformly equivalent on $\{k(x,\cdot):x\in\mbS\}$, i.e., there exist some universal constant $c,C>0$ such that 
    $$
    c\|k(x,\cdot)\|_{\mcH_k}\leq \|k(x,\cdot)\|_{L_1(P_X)}\leq C\|k(x,\cdot)\|_{\mcH_k}\quad \forall x\in\mbS.
    $$
\end{lemma}
\begin{proof}
    Denote $k_x:=k(x,\cdot)\in\mcH_k$. 
    For the upper bound, the reproducing property along with Cauchy-Schwarz inequality  
    $$
    |k(x,y)|=|\langle k_x,k_y\rangle_{\mcH_k}|\leq\|k_x\|_{\mcH_k}\|k_y\|_{\mcH_k}= \sqrt{k(x,x)}\sqrt{k(y,y)}.
    $$
    The integration with respect to $y$ gives 
    $$
    \|k_x\|_{L_1(P_X)}=\int_{\mbS} |k(x,y)|\dd P_Y\leq\sqrt{k(x,x)}\int_{\mbS} \sqrt{k(y,y)}\dd P_Y.
    $$
    Note that since $k$ is continuous and $\mbS$ is compact, we have
    $$
    C:=\int_{\mbS} \sqrt{k(y,y)}<\infty,
    $$
    which follows that 
    $$
    \|k_x\|_{L_1(P_X)}\leq C\sqrt{k(x,x)}=C\|k_x\|_{\mathcal H_k}.
    $$

    For the lower bound, we consider the ratio function 
    $$
    \Phi(x):=\frac{\|k_x\|_{L_1(P_Y)}}{\|k_x\|_{\mcH_k}}=\frac{\int_{\mbS} |k(x,y)|\dd P_Y}{\sqrt{k(x,x)}}. 
    $$
    It can be verified that $\Phi(x)$ is continuous on $\mbS$ and positive almost everywhere since $k$ is continuous on a compact set $\mbS\times\mbS$. Therefore, 
    $$
    c:=\min_{x\in\mbS}\Phi(x)>0, 
    $$
    which completes the proof. 
\end{proof}

\begin{lemma}\label{app_le_02}
    Suppose that $E|\epsilon|<\infty$ and $\sup_{n\geq1}|m_{\mcD_n}-m|\in L_{\infty}(P_X)$ a.e.\ hold. Suppose $\mcH_k\cap \mcF_n(\omega)$ is dense in $\mcF_n(\omega)$ for fixed sample point with respect to $L_1$-norm. Then \eqref{app_eq_11} holds if and only if \eqref{app_eq_12} holds. 
\end{lemma}
\begin{proof}
    We define a map 
    $$
    \Phi_n:L_1(P_X)\to\mbR,\quad \Phi(f_{\mcD_n(\omega)}):=E\big[\varepsilon_n\{f_{\mcD_n}(X)-E(f_{\mcD_n}(X)|\mcD_n)\}|\mcD_{n}\big](\omega). 
    $$
    $\Phi_n$ is well-defined by the proof of Lemma~\ref{lemma_01}. It can be verified that $\Phi_n$ is a continuous linear operator in $L_1(P_X)$ since the expectation is linear. We now prove that, for any fixed $n$ and fixed sample point almost everywhere, it holds 
    \begin{equation}\label{app_eq_13}
        \sup_{f_{\mcD_n}\in\mcF_n}\Phi_n(f_{\mcD_n(\omega)}) \leq \sup_{f_{\mcD_n(\omega)}\in\mcH_k\cap\mcF_n(\omega)}\Phi_n(f_{\mcD_n(\omega)})
    \end{equation}
    For any $f_{\mcD_n}\in\mcF_n$ and fixed sample point $\omega$, there exist a sequence $f_{m,\mcD_n(\omega)}\in \mcH_k\cap\mcF_n(\omega)$ such that $\|f_{\mcD_n(\omega)}-f_{m,\mcD_n(\omega)}\|_{L_1(P_X)}\to0$. 
    By the continuity, it follows that $\Phi_n(f_{m,\mcD_n(\omega)})\to\Phi_n(f_{\mcD_n(\omega)})$. Then, it is followed by 
    $$
    |\Phi_n(f_{\mcD_n(\omega)})|=\lim_{m\to\infty}|\Phi_n(f_{m,\mcD_n(\omega)})|\leq \sup_{g(\omega,\cdot)\in\mcH_k\cap\mcF_n(\omega)}|\Phi_n(g(\omega,\cdot))|.
    $$
    Taking the supreme with respect to $f_{\mcD_n}\in\mcF_n$, we have \eqref{app_eq_13}. Note that since the the other direction of the inequality is obvious, for any fixed $n$ and fixed sample point, it holds that
    $$
    \sup_{f_{\mcD_n}\in\mcF_n}|\Phi_n(f_{\mcD_n(\omega)})| = \sup_{f_{\mcD_n(\omega)}\in\mcH_k\cap\mcF_n(\omega)}|\Phi_n(f_{\mcD_n(\omega)})|.
    $$
    The proof is completed by noting that 
    $$
    \mcH_k\cap\mcF_n(\omega)\subseteq \{f_{\mcD_n(\omega)}\in\mcH_k\mid \|f_{\mcD_n(\omega)}\|_{L_1(P_X)}\leq1\ \text{for every } \omega\}
    \subseteq\mcF_n(\omega).
    $$
\end{proof}

\subsection{Proof of Theorem~\ref{th_02}}\label{app_sec_023}
\begin{proof}
    The proof for the null hypothesis includes two steps. First, we calculate the difference between $T_k$ and $\hat T_k$, where  
    $$
    \hat T_k:=\frac{1}{n_2(n_2-1)}\sum_{i\neq j}^{n_2}\{Y_{2,i}-m(X_{2,i})\}\{Y_{2,j}-m(X_{2,j})\}\hat U(X_{2,i},X_{2,j}), 
    $$    
    Second, we consider the difference between $\hat T_k$ and $\tilde T_k$ with   
    $$
    \tilde T_k:=\frac{1}{n_2(n_2-1)}\sum_{i\neq j}^{n_2}\{Y_{2,i}-m(X_{2,i})\}\{Y_{2,j}-m(X_{2,j})\} U(X_{2,i},X_{2,j}),
    $$
    where $U(X,X'):= k(X,X')-E_X(k(X,X'))-E_{X'}(k(X,X'))+E_{XX'}(k(X,X'))$. 
    The proof is completed by deriving the asymptotic distribution of $\tilde T_k$. 

    \textbf{Step 1. ($T_k$ and $\hat T_k$).}
    A direct calculation yields that 
    \begin{equation}\label{app_eq_17}
        \begin{aligned}
        T_k &= \hat T_k + \frac{1}{n_2(n_2-1)}\sum_{i\neq j}^{n_2}\Delta_{n_1}(X_{2,i})\{Y_{2,j}-m(X_{2,j})\}\hat U(X_{2,i},X_{2,j}) \\
        &\quad\ + \frac{1}{n_2(n_2-1)}\sum_{i\neq j}^{n_2}\Delta_{n_1}(X_{2,j})\{Y_{2,i}-m(X_{2,i})\}\hat U(X_{2,i},X_{2,j})\\
        &\quad\ + \frac{1}{n_2(n_2-1)}\sum_{i\neq j}^{n_2}\Delta_{n_1}(X_{2,i})\Delta_{n_1}(X_{2,j})\hat U(X_{2,i},X_{2,j})\\
        &=: \hat T_k + B_1 + B_2 + B_3. 
    \end{aligned}
    \end{equation}
    Next we calculate $B_1,B_2,B_3$ in order. 
    
    For simplicity, we denote $\epsilon_{2,i}=Y_{2,i}-m(X_{2,i})$, $i=1,\dots,n_2$.  
    Note that since $E(\epsilon_{2,j}|X_{2,i},X_{2,j},\mcD_{n_1})=E(\epsilon_{2,j}|X_{2,j})=0$, it follows that 
    $$
    E(B_1)=E\{E(\tilde B_1|X_{2,i},X_{2,j},\mcD_{n_1})\}=E\{\Delta_{n_1}(X_{2,i})\hat U(X_{2,i},X_{2,j})E(\epsilon_{2,j}|X_{2,i},X_{2,j},\mcD_{n_1})\}=0.
    $$
    For the second moment, it can be bounded by 
    \begin{equation}\label{app_eq_14}
        \begin{aligned}
        E(B_1^2) &= \frac{1}{n_2^2(n_2-1)^2}E\bigg[\bigg\{\sum_{i\neq j}^{n_2}\Delta_{n_1}(X_{2,i})\epsilon_{2,j}\hat U(X_{2,i},X_{2,j})\bigg\}^2\bigg]\\
        &= \frac{1}{n_2^2(n_2-1)^2}E\bigg\{\sum_{i\neq j}^{n_2}\Delta_{n_1}(X_{2,i})^2\epsilon_{2,j}^2\hat U(X_{2,i},X_{2,j})^2\\
        &\ \quad +\sum_{j\neq i_1\neq i_2}^{n_2}\Delta_{n_1}(X_{2,i_1})\Delta_{n_1}(X_{2,i_2})\epsilon_{2,j}^2\hat U(X_{2,i_1},X_{2,j})\hat U(X_{2,i_2},X_{2,j})\bigg\}\\
        &\lesssim E\bigg\{{n_2^{-2}}\Delta_{n_1}(X_{2,i})^2\hat U(X_{2,i},X_{2,j})^2\\
        &\ \quad +{n_2^{-1}}\Delta_{n_1}(X_{2,i_1})\Delta_{n_1}(X_{2,i_2})\hat U(X_{2,i_1},X_{2,j})\hat U(X_{2,i_2},X_{2,j})\bigg\},
        \end{aligned}
    \end{equation} 
    where the last inequality uses 
    \begin{align*}
        E\{\Delta_{n_1}(X_{2,i})^2\epsilon_{2,j}^2\hat U(X_{2,i},X_{2,j})^2\}
    &=E[E\{\Delta_{n_1}(X_{2,i})^2\epsilon_{2,j}^2\hat U(X_{2,i},X_{2,j})^2|X_{2,i},X_{2,j},\mcD_{n_1}\}]\\
    &=E\{\Delta_{n_1}(X_{2,i})^2\hat U(X_{2,i},X_{2,j})^2E(\epsilon_{2,j}^2|X_{2,j})\}\\
    &\lesssim E\{\Delta_{n_1}(X_{2,i})^2\hat U(X_{2,i},X_{2,j})^2\}, 
    \end{align*}
    and similar technique for $E\{\Delta_{n_1}(X_{2,i_1})\Delta_{n_1}(X_{2,i_2})\epsilon_{2,j}^2\hat U(X_{2,i_1},X_{2,j})\hat U(X_{2,i_2},X_{2,j})\}$. 
    Under the null hypothesis, H\"older's inequality gives that
    $$
    E\{\Delta_{n_1}(X_{2,i})^2\hat U(X_{2,i},X_{2,j})^2\}\leq r_{n_1}^2E\{\hat U(X_{2,i},X_{2,j})^2\},
    $$
    and 
    \begin{align*}
        &\ \quad E\{\Delta_{n_1}(X_{2,i_1})\Delta_{n_1}(X_{2,i_2})\hat U(X_{2,i_1},X_{2,j})\hat U(X_{2,i_2},X_{2,j})\}\\
        &\leq E|\Delta_{n_1}(X_{2,i_1})\Delta_{n_1}(X_{2,i_2})\hat U(X_{2,i_1},X_{2,j})\hat U(X_{2,i_2},X_{2,j})|\\
        &\leq r_{n_1}^2E\{\hat U(X_{2,i},X_{2,j})^2\}.
    \end{align*}
    For this direction, we claim that 
    \begin{equation}\label{app_eq_43}
        E\{\hat U(X_{2,i},X_{2,j})^2\}=O(1). 
    \end{equation}
    In fact, denote $\mu(x):=E(k(x,X))$ and $\eta:=E(k(X,X'))$, which follows that $U(x,x')=k(x,x')-\mu(x)-\mu(x')+\eta$. Let $\xi_i:= n_1^{-1}\sum_{l=1}^{n_1}\{k(X_{2,i},X_{1,l})-\mu(X_{2,i})\}$ and $\zeta:= n_1^{-1}(n_1-1)^{-1}\sum_{l\neq m}^{n_1}\{k(X_{1,m},X_{1,l})-\eta\}$. 
    A direct calculation leads to 
    $$
    E(\xi_i)=E\{E(\xi_i|X_{2,i})\}= 0, 
    $$
    and 
    $$
    E(\xi_i^2)=E\{E(\xi_i^2|X_{2,i})\}=\frac{1}{n_1}E\{\Var(k(X_{2,i},X_{1,1}|X_{2,i}))\}\leq\frac{1}{n_1}E\{k(X_{2,i},X_{1,1})^2\} \lesssim n_1^{-1}. 
    $$
    Another similar calculation yields $E(\zeta)=0$ and $E(\zeta^2)=O(n_1^{-2})$. 
    Note that since 
    $$
    E\{U(X_{2,i_1},X_{2,j})U(X_{2,i_2},X_{2,j})\} = E[E\{U(X_{2,i_1},X_{2,j})|X_{2,j}\}\{ U(X_{2,i_2},X_{2,j})|X_{2,j}\}]=0, 
    $$
    and $\hat U(X_{2,i},X_{2,j}) = U(X_{2,i},X_{2,j}) - \xi_i -\xi_j + \zeta$, it holds that \eqref{app_eq_43}. Combining it with \eqref{app_eq_14}, it follows that $E(B_1^2)=O(r_{n_1}^2n_2^{-1})$, and then
    \begin{equation}\label{app_eq_15}
        B_1=O(r_{n_1}n_2^{-1/2}).  
    \end{equation}
    It also applies to $B_2$ by the symmetry. 
    
    Next, we consider the upper bound for $B_3$. A use of the triangle inequality and Cauchy-Schwarz inequality shows that
    \begin{align*}
        E|B_3|&\leq E|\Delta_{n_1}(X_{2,i})\Delta_{n_1}(X_{2,j})\hat U(X_{2,i},X_{2,j})|\\
        &\leq [E\{\Delta_{n_1}(X_{2,i})^2\Delta_{n_1}(X_{2,j})^2\}E\{\hat U(X_{2,i},X_{2,j})^2\}]^{1/2}\\
        &=O(r_{n_1}^2),
    \end{align*}
    where the last inequality uses \eqref{app_eq_43}. It is followed by  
    \begin{equation}\label{app_eq_16}
        B_3=O_p(r_{n_1}^2). 
    \end{equation}

    Combining \eqref{app_eq_17}, \eqref{app_eq_15} and \eqref{app_eq_16} together, we have
    \begin{equation}\label{app_eq_18}
        T_k=\hat T_k + O_p\big(r_{n_1}n_2^{-1/2}+r_{n_1}^2\big). 
    \end{equation}

    \textbf{Step 2. ($\hat T_k$ and $\tilde T_k$).} By the orthogonality of the Hoeffding decomposition, it holds that 
    \begin{equation}\label{app_ep_19}
        E\{\hat U(X_{2,i},X_{2,j})-U(X_{2,i},X_{2,j})\}^2=E(-\xi_i-\xi_j+\zeta)^2=E(\xi_i^2+\xi_j^2+\zeta^2)=O(n_1^{-1}). 
    \end{equation}
    It follows that
    \begin{align*}
        E\{(\hat T_k-\tilde T_k)^2\}&\lesssim \frac{1}{n_2^4}E\bigg[\sum_{i\neq j}^{n_2}\epsilon_{2,i}\epsilon_{2,j}\{U(X_{2,i},X_{2,j})-\hat U(X_{2,i},X_{2,j})\}\bigg]^2\\
        &\overset{(i)}{\lesssim} \frac{1}{n_2^4}E\bigg[\sum_{i\neq j}^{n_2}\epsilon_{2,i}^2\epsilon_{2,j}^2\{U(X_{2,i},X_{2,j})-\hat U(X_{2,i},X_{2,j})\}^2\bigg]\\
        &\overset{(ii)}{\lesssim} \frac{1}{n_2^2}E[\{U(X_{2,i},X_{2,j})-\hat U(X_{2,i},X_{2,j})\}^2]\\
        &\overset{(iii)}{=} O(n_1^{-1}n_2^{-2}),
    \end{align*}
    where Step $(i)$ uses the fact that 
    \begin{align*}
        &E\big[\epsilon_{2,i_1}\epsilon_{2,j_1}\epsilon_{2,i_2}\epsilon_{2,j_2}\{U(X_{2,i_1},X_{2,j_1})-\hat U(X_{2,i_1},X_{2,j_1})\}\{U(X_{2,i_2},X_{2,j_2})-\hat U(X_{2,i_2},X_{2,j_2})\}\big]\\
        =\;& E\big[E\big[\epsilon_{2,i_1}\epsilon_{2,j_1}\epsilon_{2,i_2}\epsilon_{2,j_2}\{U(X_{2,i_1},X_{2,j_1})-\hat U(X_{2,i_1},X_{2,j_1})\}\{U(X_{2,i_2},X_{2,j_2})-\hat U(X_{2,i_2},X_{2,j_2})\}|X_{2,i_1}\big]\big]\\
        =\;& E\big[\epsilon_{2,j_1}\epsilon_{2,i_2}\epsilon_{2,j_2}\{U(X_{2,i_1},X_{2,j_1})-\hat U(X_{2,i_1},X_{2,j_1})\}\{U(X_{2,i_2},X_{2,j_2})-\hat U(X_{2,i_2},X_{2,j_2})\}E[\epsilon_{2,i_1}|X_{2,i_1}]\big]
        =0, 
    \end{align*}
    when $i_1\neq i_2$ and $i_1\neq j_2$, and Step $(ii)$ holds since 
    \begin{align*}
       & E[\epsilon_{2,i}^2\epsilon_{2,j}^2\{U(X_{2,i},X_{2,j})-\hat U(X_{2,i},X_{2,j})\}^2]\\
    =\;&E\big[E[\epsilon_{2,i}^2\epsilon_{2,j}^2\{U(X_{2,i},X_{2,j})-\hat U(X_{2,i},X_{2,j})\}^2|X_{2,i},X_{2,j}]\big]\\
    =\;&E\big[\{U(X_{2,i},X_{2,j})-\hat U(X_{2,i},X_{2,j})\}^2E(\epsilon_{2,i}^2|X_{2,i})E(\epsilon_{2,j}^2|X_{2,j})\big]\\
    \lesssim\;&E\big[\{U(X_{2,i},X_{2,j})-\hat U(X_{2,i},X_{2,j})\}^2\big]. 
    \end{align*}
    Moreover, Step $(iii)$ follows by \eqref{app_ep_19}. Therefore, a use of Markov's inequality yields
    \begin{equation}\label{app_eq_20}
        \hat T_k=\tilde T_k + O_p(n_1^{-1/2}n_2^{-1}). 
    \end{equation}
Combining \eqref{app_eq_18} and \eqref{app_eq_20}, we have 
    $$
    n_2T_k=n_2\tilde T_k + O_p\big(r_{n_1}n_2^{1/2}+ r_{n_1}^2n_2\big) = n_2\tilde T_k + o_p(1).
    $$

    Define
\[
    h(z,z'):=\{y-m(x)\}\{y'-m(x')\}U(x,x'),
\]
where $z=(x,y),z'=(x',y')$. Define the integral operator $\mcT_h:L_2(P_{XY})\to L_2(P_{XY})$ by
\[
    (\mcT_h f)(z):=\int h(z,z')f(z')\,\dd P_{XY}(z').
\]
Next, we prove that $\mcT_h$ is a compact self-adjoint operator on $L_2(P_{XY})$.

Let $Z=(X,Y)$ and $Z'=(X',Y')$ be i.i.d.\ from $P_{XY}$. By a direct calculation, we have
\[
\begin{aligned}
    E\{h(Z,Z')^2\}&=E\left[E(\epsilon^2|X)E\{(\epsilon')^2| X'\}U(X,X')^2\right] \\
    &\leq C^2 E\{U(X,X')^2\}<\infty.
\end{aligned}
\]
Hence $\mcT_h$ is a Hilbert-Schmidt integral operator on $L_2(P_{XY})$, with
\[
    \|\mcT_h\|_{\mathrm{HS}}^2
    =\iint h(z,z')^2\,\dd P_{XY}(z)\dd P_{XY}(z')<\infty.
\]
Consequently, $\mcT_h$ is compact. Since $U$ is symmetric, $h(z',z)=h(z,z')$. For any $f,g\in L_2(P_{XY})$,
\[
\begin{aligned}
    &\ \quad\iint
    |h(z,z')f(z')g(z)|
    \,\dd P_{XY}(z')\dd P_{XY}(z) \\
    &\leq
    \|h\|_{L_2(P_{XY}\otimes P_{XY})}\|f\|_{L_2(P_{XY})}\|g\|_{L_2(P_{XY})}<\infty.
\end{aligned}
\]
Fubini's theorem therefore gives
\[
\begin{aligned}
    \langle \mcT_h f,g\rangle_{L_2(P_{XY})}
    &=
    \iint h(z,z')f(z')g(z)
    \,\dd P_{XY}(z')\dd P_{XY}(z)\\
    &=\iint h(z',z)f(z')g(z)\,\dd P_{XY}(z)\dd P_{XY}(z') \\
    &=\langle f,\mcT_h g\rangle_{L_2(P_{XY})}.
\end{aligned}
\]
Therefore $\mcT_h$ is self-adjoint. Hence, $\mcT_h$ is a compact self-adjoint operator on $L_2(P_{XY})$.
    
Note that since $\tilde T_k$ is a degenerate U-statistic of order two, a direct use of the asymptotic theory for degenerate U-statistic \citep[Section~3.2]{Lee1990UStatisticsTA} yields that
    $$
    n_2\tilde T_k \overset{d}{\to}\sum_{r=1}^\infty\lambda_r(Z_r^2-1),
    $$
    where $Z_r$ are independent standard Gaussian random variables, and $\lambda_r$ are eigenvalues of the operator $\mcT_h$; that is, there exists an orthonormal basis $\{g_r\}_{r=1}^{\infty}$ of $L_2(P_{XY})$ such that
    $$
    \lambda_rg_r(x,y)=\int\{y-m(x)\}\{Y-m(X)\}U(x,X)g_r(X,Y)\;\dd P_{XY}, 
    $$
    which completes the proof. 
\end{proof}

\subsection{Proof of Theorem~\ref{th_power}}\label{app_sec_026}
\begin{proof}
Consider the operator $\mathcal{T}_U$ induced by the kernel $U(X,X')$ defined on $L_2(P_X)$ as 
    $$
    (\mathcal{T}_Uf)(x):=\int U(x,x')f(x')\dd P_X(x').
    $$
Since $k$ is positive definite, the Cauchy-Schwarz inequality gives
\[
    |k(x,x')|^2\leq k(x,x)k(x',x')
\]
for all $x,x'\in\mbS$. Hence, for independent $X,X'\sim P_X$,
\[
\begin{aligned}
    E\{k(X,X')^2\}\leq
    E\{k(X,X)k(X',X')\} =
    \{E k(X,X)\}^2
    <\infty.
\end{aligned}
\]
Moreover, a direct calculation yields
$$    
E\{U(X,X')^2\}\leq16\{E k(X,X)\}^2<\infty.
$$
Therefore, $U\in L_2(P_X\otimes P_X)$. It follows that $\mcT_U$ is a Hilbert-Schmidt integral operator on $L_2(P_X)$, with
\[
    \|\mcT_U\|_{\mathrm{HS}}^2=\iint U(x,x')^2\,\dd P_X(x)\dd P_X(x')<\infty.
\]
Hence $\mcT_U$ is compact. Since $k$ is symmetric, $U(x,x')=U(x',x)$.
For any $f,g\in L_2(P_X)$, by the Cauchy-Schwarz inequality,
\[
\begin{aligned}
    \iint
    |U(x,x')f(x')g(x)|
    \,\dd P_X(x')\dd P_X(x) \leq
    \|U\|_{L_2(P_X\otimes P_X)}
    \|f\|_{L_2(P_X)}\|g\|_{L_2(P_X)}<\infty.
\end{aligned}
\]
Thus Fubini's theorem applies, and
\[
\begin{aligned}
    \langle \mcT_U f,g\rangle_{L_2(P_X)}
    &=\iint U(x,x')f(x')g(x)\,\dd P_X(x')\dd P_X(x) \\
    &=\iint U(x',x)f(x')g(x)\,\dd P_X(x)\dd P_X(x') \\
    &=\langle f,\mcT_U g\rangle_{L_2(P_X)}.
\end{aligned}
\]
Therefore $\mcT_U$ is a compact self-adjoint Hilbert-Schmidt operator. Then, the spectral decomposition is given by
    $$
    U(x,x') = \sum_{j=1}^{\infty} \lambda_{Uj}\phi_{Uj}(x)\phi_{Uj}(x'),
    $$
    where $\{\phi_{Uj}\}_{j=1}^{\infty}$ is an orthonormal basis of $L_2(P_{X})$, and eigenvalues $\lambda_{U1}\geq\lambda_{U2}\geq\dots>0$. 
    Denote $\Delta_{n_1}=m-m_{\mcD_{n_1}}$, $b_{n_1}=E(\Delta_{n_1}(X)|\mcD_{n_1})$ and $\Delta_{n_1}^c = \Delta_{n_1}-b_{n_1}$. The centered error has the orthogonal decomposition
$$
\Delta_{n_1}^c=\sum_{j\geq1}a_{j,n_1}\phi_{Uj},\quad  a_{j,n_1}=\langle{\Delta_{n_1}^c},{\phi_{Uj}}\rangle_{L_2(P_X)}.
$$
Define $\delta_{n_1}^2:=\|\Delta_{n_1}\|^2_{L_2(P_X)}=E\{\Delta_{n_1}(X)^2|\mcD_{n_1}\}$. 

    The proof is divided into three steps: Step~1. Bounding the discrepancy between the conditional means $E(T_k|\mcD_{n_1})$ and 
    $$
    S_c(\Delta):=E\{\Delta_{n_1}(X)\Delta_{n_1}(X')U(X,X')|\mcD_{n_1}\},
    $$
    where $U(X,X')= k(X,X')-E_X(k(X,X'))-E_{X'}(k(X,X'))+E_{XX'}(k(X,X'))$. 
    Step~2. Calculating the lower bound of $S_c(\Delta)$. 
    Step~3. Bounding the conditional variance $\Var(T_k|\mcD_{n_1})$. 

    \textbf{Step~1. (Bounding the discrepancy between the conditional means).} To begin with, a direct calculation shows
    \begin{equation}\label{app_eq_32}
        \begin{aligned}
        &\ \quad E(T_k|\mcD_{n_1})-S_c(\Delta)\\
        &=E\{\Delta_{n_1}(X)\Delta_{n_1}(X')(\hat U(X,X')-U(X,X'))|\mcD_{n_1}\}\\
        &=E[\{\Delta_{n_1}(X)-b_{n_1}\}\{\Delta_{n_1}(X')-b_{n_1}\}(\hat U(X,X')-U(X,X'))|\mcD_{n_1}]\\
        &\ \quad+ 2b_{n_1}E[\{\Delta_{n_1}(X)-b_{n_1}\}(\hat U(X,X')-U(X,X'))|\mcD_{n_1}]\\
        &\ \quad+ b_{n_1}^2E\{\hat U(X,X')-U(X,X')|\mcD_{n_1}\},
    \end{aligned}
    \end{equation}
    where $X,X'$ are independent of $\mcD_{n_1}$. Note that 
    \begin{align*}
        \hat U(X,X')-U(X,X')&=E_X(k(X,X'))-\frac{1}{n_1}\sum_{l=1}^{n_1}k(X',X_{1,l})+E_{X'}(k(X,X'))-\frac{1}{n_1}\sum_{l=1}^{n_1}k(X,X_{1,l})\\
    &\ \quad+\frac{1}{n_1(n_1-1)}\sum_{l\neq m}^{n_1}k(X_{1,m},X_{1,l})-E_{XX'}(k(X,X')), 
    \end{align*}
    which follows that 
    \begin{equation}\label{app_eq_33}
        E[\{\Delta_{n_1}(X)-b_{n_1}\}\{\Delta_{n_1}(X')-b_{n_1}\}(\hat U(X,X')-U(X,X'))|\mcD_{n_1}]=0.
    \end{equation}
    Note also that
    \begin{align*}
        &\ \quad\hat U(X,X')-U(X,X')\\
        &=E_{X'}(k(X,X'))-\frac{1}{n_1}\sum_{m=1}^{n_1}k(X,X_{1,m})+\frac{1}{n_1}\sum_{m=1}^{n_1}E_{X}(k(X,X_{1,m}))-E_{XX'}(k(X,X'))\\
    &\ \quad +E_X(k(X,X'))-\frac{1}{n_1}\sum_{l=1}^{n_1}k(X',X_{1,l})+\frac{1}{n_1(n_1-1)}\sum_{l\neq m}^{n_1}\bigg\{k(X_{1,m},X_{1,l})-E_{X}(k(X,X_{1,m}))\bigg\}\\
    &=:U_1+ U_2. 
    \end{align*}
    Then, Cauchy-Schwarz inequality with some calculations yields that 
    \begin{align*}
        &\ \quad E[\{\Delta_{n_1}(X)-b_{n_1}\}U_1|\mcD_{n_1}] \\
        &= E(\Delta_{n_1}(X)U_1|\mcD_{n_1}) \\
        &= E\bigg(\Delta_{n_1}(X)\bigg\langle E(k(X,\cdot))-k(X,\cdot),\frac{1}{n_1}\sum_{l=1}^{n_1}k(X_{1,l},\cdot)-E(k(X',\cdot))\bigg\rangle_{\mcH_k}|\mcD_{n_1}\bigg)\\
        &\leq \|E[\Delta_{n_1}(X)\{E(k(X,\cdot))-k(X,\cdot)\}|\mcD_{n_1}]\|_{\mcH_k}\bigg\|\frac{1}{n_1}\sum_{l=1}^{n_1}k(X_{1,l},\cdot)-E(k(X',\cdot))\bigg\|_{\mcH_k}
    \end{align*}
    An elementary calculation gives that
    $$
    S_c(\Delta)=\|E[\Delta_{n_1}(X)\{E(k(X,\cdot))-k(X,\cdot)\}|\mcD_{n_1}]\|_{\mcH_k}^2,
    $$
    and 
    $$
    E\bigg\|\frac{1}{n_1}\sum_{l=1}^{n_1}k(X_{1,l},\cdot)-E(k(X',\cdot))\bigg\|_{\mcH_k}^2=O(n_1^{-1}).
    $$
    Another direct calculation yields that
    $$
    E[\{\Delta_{n_1}(X)-b_{n_1}\}U_2|\mcD_{n_1}]=0.
    $$
    Combining the above results together, we have 
    \begin{equation}\label{app_eq_355}
        E[\{\Delta_{n_1}(X)-b_{n_1}\}(\hat U(X,X')-U(X,X'))|\mcD_{n_1}]\leq\sqrt{S_c(\Delta)}\bigg\|\frac{1}{n_1}\sum_{l=1}^{n_1}k(X_{1,l},\cdot)-E(k(X',\cdot))\bigg\|_{\mcH_k}.
    \end{equation}
    Moreover, it also holds that
    \begin{equation}\label{app_eq_34}
    \begin{aligned}
        &\ \quad E\{\hat U(X,X')-U(X,X')|\mcD_{n_1}\}\\
        &=\frac{1}{n_1(n_1-1)}\sum_{l\neq m}^{n_1}k(X_{1,m},X_{1,l})+E_{XX'}(k(X,X'))-\frac{2}{n_1}\sum_{l=1}^{n_1}E_X(k(X,X_{1,l}))=O_p(n_1^{-1}).
    \end{aligned}
    \end{equation}
    Combining \eqref{app_eq_32}--\eqref{app_eq_34} together, we have
    \begin{align*}
        |E(T_k|\mcD_{n_1})-S_c(\Delta)|&\leq 2|b_{n_1}||E\{\Delta_{n_1}(X)(\hat U(X,X')-U(X,X'))|\mcD_{n_1}\}|\\
        &\ \quad+ b_{n_1}^2|E\{\hat U(X,X')-U(X,X')|\mcD_{n_1}\}|,
    \end{align*}
    and 
    \begin{equation}\label{app_eq_35}
        \frac{|E(T_k|\mcD_{n_1})-S_c(\Delta)|}{V_{1,n_1}}=O_p(1),
    \end{equation}
    where $V_{1,n_1}=|b_{n_1}|\sqrt{S_c(\Delta)}n_1^{-1/2}+b_{n_1}^2n_1^{-1}$. 

    \textbf{Step~2. (Calculating the lower bound of $S_c(\Delta)$).}  
    For this direction, a calculation gives
    \begin{align*}
        S_c(\Delta_{n_1})
        &=\iint \Delta_{n_1}(x)\Delta_{n_1}(z)U(x,z)\,\dd P_X(x)\dd P_X(z)\\
        &=\int\Delta_{n_1}(x)\bigg\{\int\Delta_{n_1}(z)U(x,z)\,\dd P_X(z)\bigg\}\dd P_X(x)\\
        &=\langle\Delta_{n_1},{\mcT_U\Delta_{n_1}}\rangle_{L_2(P_X)}.
    \end{align*}
    Since $\Delta_{n_1}(X)=b_{n_1}+\Delta_{n_1}^c(X)$ and $\mcT_Ub_{n_1}=0$, and the operator $\mcT_U$ is self-adjoint, it holds that
    $$
    \langle\Delta_{n_1},{\mcT_U\Delta_{n_1}}\rangle_{L_2(P_X)}=\langle{\Delta_{n_1}^c},{\mcT_U\Delta_{n_1}^c}\rangle_{L_2(P_X)}=\sum_{j\geq1}\lambda_{Uj} a_{j,n_1}^2.
    $$
    Using the centered error has the orthogonal decomposition and Assumption~\ref{ass_03}, 
    $$
        \langle{\Delta_{n_1}^c},{\mcT_U\Delta_{n_1}^c}\rangle_{L_2(P_X)}=\sum_{j\geq1}\lambda_{Uj} a_{j,n_1}^2\geq \lambda_{UJ_{n_1}}\sum_{j\geq1}^{J_{n_1}}a_{j,n_1}^2\geq \lambda_{UJ_{n_1}}(1-\rho_{n_1})\delta_{n_1}^2.
    $$
    Summarizing the results together, we have 
    \begin{equation}\label{app_eq_36}
        S_c(\Delta_{n_1})\geq \lambda_{UJ_{n_1}}(1-\rho_{n_1})\delta_{n_1}^2.
    \end{equation}

    \textbf{Step~3. (Bounding the conditional variance).} 
    Consider the Hoeffding decomposition of $T_k$ conditional on $\mcD_{n_1}$ as 
    \begin{equation}\label{app_eq_37}
        T_k- E(T_k|\mcD_{n_1}) = \frac{2}{n_2}\sum_{i=1}^{n_2}g_n(Z_{2,i})+\frac{1}{n_2(n_2-1)}\sum_{i\neq j}^{n_2}\tilde h_n(Z_{2,i},Z_{2,j})
    \end{equation}
    where $Z_{2,i}=(X_{2,i},Y_{2,i})$, $h_n(z,z'):=(y-m_{\mcD_{n_1}}(x))(y'-m_{\mcD_{n_1}}(x'))\hat U(x,x')$, $g_n(z)=E\{h_n(z,Z)|\mcD_{n_1}\}-E\{h_n(Z,Z')|\mcD_{n_1}\}$ and $\tilde h_n(z,z')=h_n(z,z')-g_n(z)-g_n(z')-E\{h_n(Z,Z')|\mcD_{n_1}\}$, $z=(x,y)$, $z'=(x',y')$. Since the two terms in the decomposition are orthogonal conditional on $\mcD_{n_1}$, it holds that
    \begin{equation}\label{app_eq_38}
        \Var(T_k|\mcD_{n_1})=\frac{4}{n_2}E\{g_n^2(Z)|\mcD_{n_1}\} +\frac{2}{n_2(n_2-1)}E\{\tilde h_n^2(Z,Z')|\mcD_{n_1}\}.
    \end{equation}
    Next, we calculate the two terms in order. To the end, note that
    \begin{align}\label{app_eq_39}
        E\{g_n^2(Z)|\mcD_{n_1}\}\leq E[\{Y-m_{\mcD_{n_1}}(X)\}^2E_{X'}^2\{(Y'-m_{\mcD_{n_1}}(X'))\hat U(X,X')|\mcD_{n_1}\}|\mcD_{n_1}].
    \end{align}
    A calculation follows as 
    \begin{align*}
        E_{X'}\{(Y'-m_{\mcD_{n_1}}(X'))\hat U(X,X')|\mcD_{n_1}\}&=E_{X'}\{\Delta_{n_1}(X')\hat U(X,X')|\mcD_{n_1}\}\\
        &= E_{X'}\{\Delta_{n_1}(X')\hat U_p(X,X')|\mcD_{n_1}\}+b_{n_1}\{\hat U(X,X')-\hat U_p(X,X')\}, 
    \end{align*}
    where 
    $$
        \hat U_p(x,x')=k(x,x')-\frac{1}{n_1}\sum_{l=1}^{n_1}k(x,X_{1,l})-\frac{1}{n_1}\sum_{m=1}^{n_1}k(X_{1,m},x')+\frac{1}{n_1(n_1-1)}\sum_{l\neq m}^{n_1}k(X_{1,m},X_{1,l}). 
    $$
    By Cauchy-Schwarz inequality, we have 
    \begin{align*}
        &\ \quad E_{X'}\{\Delta_{n_1}(X')\hat U_p(X,X')|\mcD_{n_1}\}\\
        &\leq \bigg\|E\bigg[\Delta_{n_1}(X')\bigg\{\frac{1}{n_1}\sum_{l=1}^{n_1}k(X_{1,l},\cdot)-k(X',\cdot)\bigg\}|\mcD_{n_1}\bigg]\bigg\|_{\mcH_k}\bigg\|\frac{1}{n_1}\sum_{l=1}^{n_1}k(X_{1,l},\cdot)-k(X,\cdot)\bigg\|_{\mcH_k}.
    \end{align*}
    By the triangle inequality and Cauchy-Schwarz inequality, we have 
    \begin{align*}
        &\ \quad \bigg\|E\bigg[\Delta_{n_1}(X')\bigg\{\frac{1}{n_1}\sum_{l=1}^{n_1}k(X_{1,l},\cdot)-k(X',\cdot)\bigg\}|\mcD_{n_1}\bigg]\bigg\|_{\mcH_k}\\
        &= \bigg\|E\bigg[\Delta_{n_1}(X')\bigg\{\frac{1}{n_1}\sum_{l=1}^{n_1}k(X_{1,l},\cdot)-E(k(X',\cdot))+E(k(X',\cdot))-k(X',\cdot)\bigg\}|\mcD_{n_1}\bigg]\bigg\|_{\mcH_k}\\
        &\leq \sqrt{S_c(\Delta)}+|b_{n_1}|\bigg\|\frac{1}{n_1}\sum_{l=1}^{n_1}k(X_{1,l},\cdot)-E(k(X,\cdot))\bigg\|_{\mcH_k}
    \end{align*}
    An elementary calculation yields that 
    $$
    \hat U(X,X')-\hat U_p(X,X')=\frac{1}{n_1^2(n_1-1)}\sum_{l\neq m}^{n_1}k(X_{1,m},X_{1,l})-\frac{1}{n_1^2}\sum_{l=1}^{n_1}k(X_{1,l},X_{1,l}).
    $$ 
    Combining these calculations with \eqref{app_eq_39} together, and using Assumptions~\ref{ass_01}--\ref{ass_02}, we have
    \begin{equation}\label{app_eq_40}
        \frac{E\{g_n^2(Z)|\mcD_{n_1}\}}{{S_c(\Delta)}+{b_{n_1}^2}{n_1^{-1}
        +{b_{n_1}^2}{n_1}^{-2}}}=O_p(1). 
    \end{equation}
    Moreover, we also have 
    \begin{equation}\label{app_eq_41}
        E\{\tilde h_n^2(Z,Z')|\mcD_{n_1}\}\leq E\{h_n^2(Z,Z')|\mcD_{n_1}\}=O_p(1).
    \end{equation}
    Combining \eqref{app_eq_38}--\eqref{app_eq_41} yields the upper bound of the conditional variance as
    \begin{equation}\label{app_eq_42}
        \frac{\Var(T_k|\mcD_{n_1})}{V_{2,n_1}^2}=O_p(1),
    \end{equation}
    where 
    $$
    V_{2,n_1}^2= \frac{S_c(\Delta)}{n_2}+\frac{b_{n_1}^2}{n_1n_2}+\frac{1}{n_2^2}. 
    $$

    Finally, we combine the results in \eqref{app_eq_35}, \eqref{app_eq_36}, and \eqref{app_eq_42} to complete the proof. For any $\varepsilon>0$, there exists $C>0$ such that, for all sufficiently large $n_1$,
    $$
    P\left(\frac{\Var(T_k|\mcD_{n_1})}{V_{2,n_1}^2}>C\right)<\frac{\varepsilon}{2}.
    $$
    Hence, by conditional Chebyshev's inequality, for every $M>0$,
    \begin{align*}
    &\ \quad P\left(\frac{\left|T_k-E(T_k|\mcD_{n_1})\right|}{V_{2,n_1}}>M\right)\\
    &\leq P\left(\frac{\Var(T_k|\mcD_{n_1})}{V_{2,n_1}^2}>C\right)+E\left[\mathbbm{1}\left\{\frac{\Var(T_k|\mcD_{n_1})}{V_{2,n_1}^2}\leq C\right\}P\left(\frac{|T_k-E(T_k|\mcD_{n_1})|}{V_{2,n_1}}>M\bigg|\mcD_{n_1}\right)\right]\\
    &\leq \frac{\varepsilon}{2}+\frac{C}{M^2}.
    \end{align*}
Taking $M$ sufficiently large gives
$$
\frac{T_k-E(T_k|\mcD_{n_1})}{V_{2,n_1}}=O_p(1).
$$
Moreover, note that since 
$$
\frac{V_{1,n_1}}{S_c(\Delta)}=o_p(1),\quad\frac{V_{2,n_1}}{S_c(\Delta)}=o_p(1).
$$
It follows that 
\begin{align*}
\frac{T_k-S_c(\Delta)}{S_c(\Delta)}=\frac{T_k-E(T_k|\mcD_{n_1})}{V_{2,n_1}}\frac{V_{2,n_1}}{S_c(\Delta)}+\frac{E(T_k|\mcD_{n_1})-S_c(\Delta)}{V_{1,n_1}}\frac{V_{1,n_1}}{S_c(\Delta)}=o_p(1).
\end{align*}
Thus,
$$
\frac{T_k}{S_c(\Delta)}\overset{p}{\to}1.
$$
Finally, for every fixed $C_1>0$,
\begin{align*}
P(n_2T_k\leq C_1)\leq
P\left(
\left|\frac{T_k}{S_c(\Delta)}-1\right|>\frac12
\right)+P\left(n_2S_c(\Delta)\leq 2C_1\right),
\end{align*}
which converges to $0$ as $n_1\to\infty$. It implies
$n_2T_k\overset{p}{\to}\infty$ and thus the proof is completed. 
\end{proof}

\subsection{Proof of Theorem~\ref{th_03}}\label{app_sec_024}
\begin{proof}
    We consider 
    $$
    \tilde T_k^b:=\frac{1}{n_2(n_2-1)}\sum_{i\neq j}^{n_2}\{Y_{2,i}-m(X_{2,i})\}\{Y_{2,j}-m(X_{2,j})\} U(X_{2,i},X_{2,j})e_{bi}e_{bj}.
    $$
    A decomposition similar to that in the proof of Theorem~\ref{th_02} follows as 
    \begin{equation}\label{app_eq_25}
        T_k^b = \tilde T_k^b + B_1^b+B_2^b+B_3^b+B_4^b. 
    \end{equation}
Here,  
$$
B_1^b:= \frac{1}{n_2(n_2-1)}\sum_{i\neq j}^{n_2}\Delta_{n_1}(X_{2,i})\epsilon_{2,j}\hat U(X_{2,i},X_{2,j})e_{bi}e_{bj}, 
$$
where $\Delta_{n_1}(X_{2,i})=m(X_{2,i})-m_{\mcD_{n_1}}(X_{2,i})$ and $\epsilon_{2,j}=Y_{2,j}-m(X_{2,j})$. $B_2^b$ is defined in a similar fashion. Moreover, 
$$
B_3^b:=\frac{1}{n_2(n_2-1)}\sum_{i\neq j}^{n_2}\Delta_{n_1}(X_{2,i})\Delta_{n_1}(X_{2,j})\hat U(X_{2,i},X_{2,j})e_{bi}e_{bj}. 
$$
$B_4^b$ is defined as 
$$
B_4^b:= \frac{1}{n_2(n_2-1)}\sum_{i\neq j}^{n_2}\epsilon_{2,i}\epsilon_{2,j}\{\hat U(X_{2,i},X_{2,j})-U(X_{2,i},X_{2,j})\}e_{bi}e_{bj}. 
$$
We use $E^*$ and $P^*$ to denote the conditional expectation and probability of a random variable conditioning on $\mcD_{n}$, respectively. A simple calculation leads to 
\begin{align}\label{app_eq_22}
    E^*\big({B_1^b}^2\big) &= \frac{1}{n_2^2(n_2-1)^2}\sum_{i\neq j}^{n_2}\Delta_{n_1}(X_{2,i})\Delta_{n_1}(X_{2,j})\epsilon_{2,i}\epsilon_{2,j}\hat U(X_{2,i},X_{2,j})^2\nonumber\\
    &\quad\ +\frac{1}{n_2^2(n_2-1)^2}\sum_{i\neq j}^{n_2}\Delta_{n_1}(X_{2,i})^2\epsilon_{2,j}^2\hat U(X_{2,i},X_{2,j})^2.
\end{align}
By $E\{\hat U(X_{2,i},X_{2,j})^2\}=O(1)$ in \eqref{app_eq_43} and H\"older's inequality, it holds that 
$$
E|E^*\big({B_1^b}^2\big)|=O\big(r_{n_1}^2n_2^{-2}\big).
$$
By Markov's inequality, the right-hand side of \eqref{app_eq_22} is $O_p\big(r_{n_1}^2n_2^{-2}\big)$ almost everywhere under the null. Similarly, we have $E^*\big({B_2^b}^2\big)=O_p\big(r_{n_1}^2n_2^{-2}\big)$. Another calculation yields that
\begin{align}\label{app_eq_23}
    E^*\big({B_3^b}^2\big) &= \frac{2}{n_2^2(n_2-1)^2}\sum_{i\neq j}^{n_2}\Delta_{n_1}(X_{2,i})^2\Delta_{n_1}(X_{2,j})^2\hat U(X_{2,i},X_{2,j})^2\nonumber\\
    &= O_p\big(r_{n_1}^4n_2^{-2}\big), 
\end{align}
where the last equation uses Markov's inequality and $E\{\hat U(X_{2,i},X_{2,j})^2\}=O(1)$ in \eqref{app_eq_43}, which has been proved in proof of Theorem~\ref{th_02}. 
Moreover, we have 
\begin{align}\label{app_eq_24}
    E^*\big({B_4^b}^2\big) &= \frac{2}{n_2^2(n_2-1)^2}\sum_{i\neq j}^{n_2}\epsilon_{2,i}^2\epsilon_{2,j}^2\{\hat U(X_{2,i},X_{2,j})-U(X_{2,i},X_{2,j})\}^2\nonumber\\
    &= O_p\big(n_1^{-1}n_2^{-2}\big), 
\end{align}
where the last equation holds by \eqref{app_ep_19} along with Markov's inequality. Combining \eqref{app_eq_22}-\eqref{app_eq_24} with \eqref{app_eq_25}, we have 
$$
E^*\{(T_k^b-\tilde T_k^b)^2\} = o_p(n_2^{-2}). 
$$
By Chebyshev's inequality, the difference between $T_k^b$ and $\tilde T_k^b$ is asymptotically negligible in the sense that
\begin{equation*}
    P(P^*(|n_2T_k^b-n_2\tilde T_k^b|\geq \delta_1)\geq\delta_2)\to0,\quad \text{for any } \delta_1>0,\delta_2>0. 
\end{equation*}
The proof is completed by noting that $n_2\tilde T_k^b\overset{D^*}{\to}\sum_{r=1}^\infty\lambda_r(Z_r^2-1)$ by the theory for bootstrapping for U-statistics \citep[Theorem~3.1]{Dehling1994RandomQF}. 
\end{proof}

\subsection{Proof of Theorem~\ref{th_04}}\label{app_sec_025}
\begin{proof}
Recalling the decomposition of \eqref{app_eq_25} in the proof on Theorem~\ref{th_02}, we investigate the asymptotic properties under the alternatives by controlling $B_1^b,B_2^b,B_3^b$. 

Recalling \eqref{app_eq_22}, we have 
\begin{equation}\label{app_eq_26}
\begin{aligned}
    E\Big\{E^*\big({B_1^b}^2\big)|\mcD_{n_1}\Big\} &= \frac{1}{n_2(n_2-1)}E\Big\{\Delta_{n_1}(X_{2,i})^2\epsilon_{2,j}^2\hat U(X_{2,i},X_{2,j})^2|\mcD_{n_1}\Big\}\\
     &\overset{(i)}{\leq} C_1{n_2^{-2}}E\Big\{\Delta_{n_1}(X_{2,i})^2\hat U(X_{2,i},X_{2,j})^2|\mcD_{n_1}\Big\}\\
     &\overset{(ii)}{\leq} \frac{C_2}{n_2^{2}}E\Big\{\Delta_{n_1}(X)^2k(X,X)|\mcD_{n_1}\Big\}+\frac{C_3}{n_2^{2}}E\Big\{\Delta_{n_1}(X)^2|\mcD_{n_1}\Big\}\\
     & = O_p(n_2^{-2}), 
\end{aligned}
\end{equation}
where Step $(i)$ holds by the law of iterated exceptions with the $\sigma$-field generated by $\{\mcD_{n_1},X_{2,i},X_{2,j}\}$. Step $(ii)$ expands the term $\hat U(X_{2,i},X_{2,j})^2$, and uses that $k(x,x')^2\leq k(x,x)k(x',x')$  and $k(x,x)+1\geq 2k^{1/2}(x,x)$ along with some elementary calculations. The same calculation also applies to $B_2^b$ by the symmetry. Moreover, it also holds that 
\begin{equation}\label{app_eq_27}
    \begin{aligned}
    E\Big\{E^*\big({B_3^b}^2\big)|\mcD_{n_1}\Big\} &= \frac{2}{n_2(n_2-1)}E\Big\{\Delta_{n_1}(X_{2,i})^2\Delta_{n_1}(X_{2,j})^2\hat U(X_{2,i},X_{2,j})^2|\mcD_{n_1}\Big\}\\
    &\leq \frac{C_4}{n_2^2}E\Big\{\Delta_{n_1}(X)^2k(X,X)|\mcD_{n_1}\Big\}E\Big\{\Delta_{n_1}(X)^2|\mcD_{n_1}\Big\}\\
    &\quad\ + \frac{C_5}{n_2^2}E^2\Big\{\Delta_{n_1}(X)^2k(X,X)|\mcD_{n_1}\Big\}\\
    &\quad\ + \frac{C_6}{n_2^2}E^2\Big\{\Delta_{n_1}(X)^2|\mcD_{n_1}\Big\}\\
    &= O_p(n_2^{-2}),
\end{aligned}
\end{equation}
where the inequality holds by a similar calculation to \eqref{app_eq_26}. 

Combining \eqref{app_eq_24}, \eqref{app_eq_26}, \eqref{app_eq_27}, and $n_2\tilde T_k^b\overset{D^*}{\to}\sum_{r=1}^\infty\lambda_r(Z_r^2-1)$, we have 
$E^*(n_2^2{T_k^b}^2)=O_p(1)$. A use of the conditional Markov's inequality yields that
$$
P^*(|n_2T_k^b|>M)\leq \frac{E^*(n_2^2{T_k^b}^2)}{M^2}. 
$$
For any $M>0$ and $\delta>0$, by Markov's inequality, it holds that
\begin{equation}\label{app_eq_30}
    \limsup_{n\to\infty}P(P^*(|n_2T_k^b|>M)>\delta)<\frac{1}{M^2\delta}E\{E^*(n_2^2{T_k^b}^2)\}.
\end{equation}
Note that since $E\{E^*(n_2^2{T_k^b}^2)\}<\infty$, taking $M\to\infty$ on the both sides in \eqref{app_eq_30} yields that 
$$\lim_{M\to\infty}\limsup_{n_1,n_2\to\infty}P(P^*(|n_2T_k^b|>M)>\delta)=0,$$
which implies that $n_2T_k^b=O_{P^*}(1)$. 


To further analyze power, we consider 
$$
T_{k,\alpha}^*:= \inf\{t:P^*(n_2T_k^b\leq t)\geq1-\alpha\}. 
$$
For any $M>0$, note that since the event $\{T_{k,\alpha}^*>M\}\subseteq\{P^*(n_2T_k^b>M)>\alpha\}$, we have
\begin{equation}\label{app_eq_28}
    P(T_{k,\alpha}^*>M)\leq P(P^*(n_2T_k^b>M)>\alpha). 
\end{equation}
By the set operation and sub-additivity of the probability, for any $M>0$, it holds that 
\begin{equation}\label{app_eq_29}
P(n_2T_k>T_{k,\alpha}^*)\geq P(n_2T_k>M, T_{k,\alpha}^*\leq M) \geq 1- P(n_2T_k\leq M)-P(T_{k,\alpha}^*>M).
\end{equation}
Invoking Theorem~\ref{th_02}, for any given $\delta_3>0$, there exists $N$ such that for $n>N$, we have $P(n_2T_k\leq M)<\delta_3/2$. By \eqref{app_eq_30}, there exists $M>0$ such that 
$P(P^*(n_2T_k^b>M)>\alpha)<\delta_3/2$. Combining with \eqref{app_eq_28} and \eqref{app_eq_29}, the proof is then completed. 
\end{proof}

\section{Distributions of $T_f$ under Alternatives}\label{app_sec_04}
Under the alternatives, by the proof of Theorem~\ref{th_01} in Section~\ref{app_sec_022}, we have 
$$
\sqrt{n_2}(T_f-A_3-A_4)\overset{d}{\to}\mcN(0,V_f). 
$$
where
$$
A_3 = \frac{1}{n_2}\sum_{i=1}^{n_2}\big\{m(X_{2,i})-m_{\mcD_{n_1}}(X_{2,i})\big\}\big\{f(X_{2,i})-E(f(X))\big\}, 
$$
and 
$$
A_4=\frac{1}{n_2}\sum_{i=1}^{n_2}\big\{m(X_{2,i})-m_{\mcD_{n_1}}(X_{2,i})\big\}\big\{E(f(X))-\hat m_f\big\}.
$$
The limiting distribution of $E_f:=A_3+A_4$ can be established under the additional Lyapounov conditions, which is required in applying the central limit theorem of the triangle array. 

\begin{assumption}\label{app_ass_01}
There exists $\delta_1,\delta_2>0$ such that $n_2^{-\delta_1/2}s_V^{2+\delta_1}E(|V-E_V|^{2+\delta_1}|\mcD_{n_1})=o_p(1)$, where $V=\{m(X)-m_{\mcD_{n_1}}(X)\}\{f(X)-E(f(X))\}$, $E_V=E(V|\mcD_{n_1})$, and $s_V^2= \Var(V|\mcD_{n_1})$, and  
$n_2^{-\delta_2/2}s_U^{2+\delta_2}E(|U-E_U|^{2+\delta_2}|\mcD_{n_1})=o_p(1)$ with 
$U=m(X)-m_{\mcD_{n_1}}(X)$, $E_U=E(U|\mcD_{n_1})$, and $s_U^2= \Var(U|\mcD_{n_1})$.     
\end{assumption}

To facilitate our discussion, we denote $V_i=\{m(X_{2,i})-m_{\mcD_{n_1}}(X_{2,i})\}\{f(X_{2,i})-E(f(X))\}$, and $U_i=m(X_{2,i})-m_{\mcD_{n_1}}(X_{2,i})$ $i=1,\dots,n_2$. A calculation along with the central limit theorem of the triangle array \citep[Corollary~9.5.11]{Ocappe2005} yields that 
$$
\frac{\sqrt{n_2}}{s_V}(A_3-E_V)=\frac{\sqrt{n_2}}{s_V}\bigg(\frac{1}{n_2}\sum_{i=1}^{n_2}V_i-E_V\bigg)\overset{d}{\to}\mcN(0,1). 
$$
It follows that 
\begin{equation}\label{app_eq_08}
    A_3 = E_V + \frac{s_V}{\sqrt{n_2}}Z_3 + o_p(n_2^{-1/2}s_V),
\end{equation}
where $Z_3$ is a standard Gaussian random variable. Similarly, it holds
$$
\frac{\sqrt{n_2}}{s_U}\bigg(\frac{1}{n_2}\sum_{i=1}^{n_2}U_i-E_U\bigg)\overset{d}{\to}\mcN(0,1). 
$$
It follows that 
\begin{equation}\label{app_eq_09}
    A_4 = \bigg\{E_U + \frac{s_U}{\sqrt{n_2}}Z_4 + o_p(n_2^{-1/2}s_U)\bigg\}\bigg\{\frac{\sqrt{\Var(f(X))}}{\sqrt{n_1}}Z_5+o_p(n_1^{-1/2})\bigg\},
\end{equation}
where $Z_4$ and $Z_5$ are asymptotically independent standard Gaussian random variables, and are also independent of $Z_3$. Combining the results in \eqref{app_eq_08} and \eqref{app_eq_09} implies the asymptotic property of $E_f=A_3+A_4$, which follows that, under the alternatives, if $n_1,n_2\to\infty$, it holds that 
    $$
    \sqrt{n_2}(T_f-E_f)\overset{d}{\to}\mcN(0,V_f).
    $$
    
Assume that $E_U=0$, $s_U=O_p(1)$ and $s_V=O_p(1)$. If $m(X)-m_{\mcD_{n_1}}(X)$ is orthogonal to the projection direction $f(X)$ for almost every $\mcD_{n_1}$, then $E_V=0$ almost surely. Consequently, $E_f=O_p(1)$ and $E(E_f)=o_p(1)$, implying that the test is not consistent. Moreover, even when $m(X)-m_{\mcD_{n_1}}(X)$ is not orthogonal to $f(X)$, $E_V$ does not admit a uniform positive lower bound. This is because the alternatives impose a rate condition only on $\inf_x|m(x)-m_{\mcD_{n_1}}(x)|$, without restricting the sign of $m(x)-m_{\mcD_{n_1}}(x)$.

\section{Dimension-Agnostic Results}\label{app_sec_E}
The theoretical results for the kernel-based testing are established by setting the dimension $p$ being fixed. In this section, we allow the distribution of the observations, the dimension of $X$, and the kernel to depend on $n$. 
The derivations require additional assumptions, which are naturally satisfied in the fixed-dimensional setting. Similar conclusions are investigated in \cite{WangXu2022AnAR} and \cite{cyx2025jrssb}. 

\subsection{Limiting Distribution of the Statistic}
With a little abuse of notation, we define $Z_{2,i}:=(X_{2,i},Y_{2,i})$ for $i=1,\dots,n_2$, and 
$$
h(z,z'):=(y-m(x))(y'-m(x'))U(x,x'),
$$
where $z=(x,y)$, $z'=(x',y')$. We first focus on the asymptotic distribution of 
$$
    \tilde T_k=\frac{1}{n_2(n_2-1)}\sum_{i\neq j}^{n_2}h(Z_{2,i},Z_{2,j}).
    $$
Since $E(Y-m(X)| X)=0$, a simple calculation yields that 
$$E\{h(Z_{2,i},Z_{2,j})|Z_{2,i}\}=0\ a.e.$$ 
Let $\sigma^2(x):=E((Y-m(X))^2|X=x)$, and define 
$$
K_n(x,x'):=\sigma(x)U(x,x')\sigma(x').
$$

The associated integral operator $\mathcal{T}_K$ defined on $L_2(P_X)$ as 
    $$
    (\mathcal{T}_Kf)(x):=\int K_n(x,x')f(x')\dd P_X(x').
    $$
Assumption~\ref{ass_01} and $E(k(X,X))<\infty$ imply that
$$
\|K_n\|_{L_2(P_{X}\otimes P_{X})}^2\leq\left[E\{\sigma_n^2(X_{n,1})U_n(X_{n,1},X_{n,1})\}\right]^2<\infty,
$$
where we used $|U_n(x,x')|^2\le U_n(x,x)U_n(x',x')$. Hence $K_n\in L_2(P_{X}\otimes P_{X})$. Consequently, the integral operator $\mathcal T_n$ is compact, self-adjoint, Hilbert-Schmidt, and nonnegative. Then, there exist nonnegative eigenvalues $\kappa_{n1}\geq\kappa_{n2}\geq\dots\geq0$ and an orthonormal basis of $\{\psi_{nr}\}_{r=1}^{\infty}\subset L_2(P_{X})$ such that 
    $$
    K_n(x,x') = \sum_{r=1}^{\infty} \kappa_{nr}\psi_{nr}(x)\psi_{nr}(x'),
    $$
    and
$$
        s_n^2:=\mathbb E\{h(Z_{n,1},Z_{n,2})^2\}=\|K_n\|_{L_2(P_{X,n}\otimes P_{X,n})}^2=\sum_{r=1}^{\infty} \kappa_{nr}^2.
$$
Without loss of generality, suppose $\liminf_{n\to\infty} s_n^2>0$ and $s_n=\Omega(1)$. Let $\lambda_{nr}={\kappa_{nr}}/{s_n}$ for $r\geq1$. Then $\lambda_{n1}\ge\lambda_{n2}\ge\cdots\ge0$ and $\sum_{r\ge1}\lambda_{nr}^2=1$.
We first introduce the following spectral limits assumption. 

\begin{assumption}\label{app_ass_E1}
There exist nonnegative constants $\lambda_1,\lambda_2,\dots$ such that $\lambda_{nr}\to \lambda_r$ for $r\ge1$. 
\end{assumption}

This assumption yields the following result to measure the escaping mass. 

\begin{lemma}\label{app_lemma_E1}
    Suppose Assumption~\ref{app_ass_E1} holds. There a determinist sequence of positive integers $R_n\to\infty$ such that, as $n_2\to\infty$, 
    \begin{equation}\label{app_eq_D01}
        \sum_{r=1}^{R_n}\lambda_{nr}^2\to\sum_{r=1}^\infty \lambda_r^2.
    \end{equation}
\end{lemma}

For $r\ge1$, define
$$
        \xi_{n,i,r}=
        \begin{cases}
        \displaystyle\frac{(Y_{2,i}-m(X_{2,i}))\psi_{nr}(X_{2,i})}{\sigma(X_{2,i})},& \sigma(X_{2,i})>0,\\
        0, & \sigma(X_{2,i})=0.
        \end{cases}
$$
For the tail part and $j=1,\dots,n_2$, define
$$
        G_{n,j}=\frac{\sqrt 2}{(n_2-1)s_n}\sum_{1\leq i<j}\sum_{\ell>R_n}\kappa_{n\ell}\xi_{n,i,\ell}\xi_{n,j,\ell},
$$
with the convention that $G_{n,1}=0$. For every $a\in\mbR^R$ and $b\in\mbR$, define 
$$
        \Delta_{n,j}(a,b):= n_2^{-1/2}a^\top(\xi_{n,j,1},\dots,\xi_{n,j,R})^\top+bG_{n,j}. 
$$
Note that for every nonzero eigenvalue $\kappa_{nr}$, the corresponding eigenfunction satisfies $\psi_{nr}=0$ a.e.\ on $\{\sigma=0\}$. Therefore, $E(\xi_{n,1,r}\xi_{n,1,s})=\mathbbm 1(r=s)$ whenever the corresponding eigenvalues are nonzero. Let $\mathcal F_{n,j}=\sigma(Z_{n,1},\dots,Z_{n,j})$, $\mathcal F_{n,0}=\{\varnothing,\Omega\}$. The following Lindeberg condition is further required to restrict the tail behaviors of the eigenfunctions. 

\begin{assumption}\label{app_ass_E3}
For every fixed $R\geq1$, every $a\in\mbR^R$ and $b\in\mbR$, we have
$$
        \sum_{j=1}^{n_2}E\{\Delta_{n,j}(a,b)^2|\mathcal F_{n,j-1}\}\overset{p}{\to}\|a\|_2^2+b^2\rho^2,
$$
and, for every $\eta>0$,
$$
        \sum_{j=1}^{n_2}E\left[\Delta_{n,j}(a,b)^2\mathbbm 1\{|\Delta_{n,j}(a,b)|>\eta\}|\mathcal F_{n,j-1}\right]\overset{p}{\to}0.
$$
\end{assumption}

The joint Lindeberg condition can also be checked separately for the fixed spectral head and the tail. For the head, a convenient sufficient condition is that, for every fixed $R$ with $\lambda_R>0$, there exist $\delta>0$ and $n_R<\infty$ such that
\[
    \sup_{n\ge n_R}E\|(\xi_{n,1,1},\dots,\xi_{n,1,R})^\top\|_2^{2+\delta}<\infty.
\]
For the tail, the conditional Lyapunov condition
\[
    \sum_{j=1}^{n_2}E(|G_{n,j}|^{2+\delta}|\mcF_{n,j-1})\overset{p}{\to}0
\]
implies the tail Lindeberg condition. It follows from the martingale argument in \cite{Hall1980MartingaleLT} and is implied by ${\lambda_{n,R_n+1}^2}/{\sum_{\ell>R_n}\lambda_{n\ell}^2}\to0$. 
These two separate bounds imply the Lindeberg condition for every linear combination $\Delta_{n,j}(a,b)$, by splitting the event $\{|x+y|>\eta\}$ according as $|x|>\eta/2$ or $|y|>\eta/2$.

The following theorem establishes the limiting distribution of the statistic with a varying kernel. 

\begin{theorem}\label{app_th_01}
Suppose Assumptions~\ref{ass_01}, \ref{app_ass_E1}--\ref{app_ass_E3} hold and $E(k(X,X))<\infty$. Under the null $H_0$, if $n_1,n_2\to\infty$ and $n_2=o(r_{n_1}^{-2})$, we have 
$$
        \frac{n_2 T_k}{\sqrt 2 s_n}
        \overset{d}{\to}\frac{1}{\sqrt 2}\sum_{r=1}^\infty \lambda_r(Z_r^2-1)+\bigg(1-\sum_{r=1}^{\infty}\lambda_r^2\bigg)^{1/2} Z_0,
$$
where $Z_0,Z_1,Z_2,\dots$ are independent standard Gaussian random variables. The infinite series is well-defined in $L_2$.
\end{theorem}

\subsection{Validity of the Wild Bootstrap}
Let $e_1,\dots,e_{n_2}$ be i.i.d.\ sub-Gaussian random variables with $E(e_1)=0$ and $\Var(e_1)=1$, independent of the data, and define
$$
        \tilde T_k^*=
        \frac{1}{n_2(n_2-1)}\sum_{1\le i\ne j\le n_2}h(Z_{2,i},Z_{2,j})e_ie_j .
$$
Let
$$
        \mathbf H_n=\left(\frac{\mathbbm 1(i\ne j)h(Z_{2,i},Z_{2,j})}{n_2-1}\right)_{1\le i,j\le n_2}.
$$
Then $n_2\tilde T_k^*=e^\top\mathbf H_n e$ with $e=(e_1,\dots,e_{n_2})^\top$. 
Define
$$
        \hat s_n^2=\frac{1}{n_2(n_2-1)}\sum_{1\le i\ne j\le n_2}h(Z_{2,i},Z_{2,j})^2 .
$$
Since $\mathbf H_n$ has zero diagonal, it holds that 
$$
        \operatorname{tr}(\mathbf H_n)=0,\quad
        \|\mathbf H_n\|_F^2=\frac{n_2}{n_2-1}\hat s_n^2, 
$$
where $\|\cdot\|_F$ denotes the Frobenius norm of a matrix. Let $\hat\mu_{n1}\ge\hat\mu_{n2}\ge\dots\ge\hat\mu_{nn_2}$ be the eigenvalues of $\mathbf H_n$ in non-increasing order, and let $\hat v_{n1},\dots,\hat v_{nn_2}$ be corresponding orthonormal eigenvectors. 
be a spectral decomposition. On the event $\widehat s_n>0$, we denote $\hat\lambda_{n\ell}={\hat\mu_{n\ell}}/{\hat s_n}$.
We first introduce two assumptions required in the validity of the wild bootstrap. 

\begin{assumption}\label{app_ass_E4}
${\hat s_n}/{s_n}\overset{p}{\to}1$ and $\hat\lambda_{nr}\overset{p}{\to}\lambda_r$ for fixed $r\ge1$. 
\end{assumption}

Assumption~\ref{app_ass_E4} entails the convergence of empirical spectral approximations. The Hoeffding decomposition shows that the sufficient conditions
\[
    \frac{E\{E\{h_n(Z_{n,1},Z_{n,2})^2|Z_{n,1}\}^2\}}{n_2s_n^4}\to0,
    \quad
    \frac{E\{h_n(Z_{n,1},Z_{n,2})^4\}}{n_2^2s_n^4}\to0
\]
imply $\hat s_n/s_n\to_p1$. Moreover, a sufficient condition for the empirical spectral approximation $\hat\lambda_{nr}\overset{p}{\to}\lambda_r$ is the empirical spectrum converges to the population one in $L_2$ metric, which particularly holds for fixed Hilbert–Schmidt kernel; see \cite{KoltchinskiiGine2000}. This assumption indicates the following result. 

\begin{lemma}\label{app_lemma_E2}
    Suppose Assumption~\ref{app_ass_E4} holds. There a determinist sequence of positive integers $J_n\to\infty$ with $J_n\leq n_2$ such that, as $n_2\to\infty$, 
    \begin{equation}\label{app_eq_D02}
        \sum_{\ell=1}^{J_n}\hat\lambda_{n\ell}^2\overset{p}{\to}
        \sum_{\ell=1}^\infty\lambda_\ell^2,\quad
        \sum_{\ell>J_n}^{n_2}\hat\lambda_{n\ell}^2\overset{p}{\to}\rho^2.
    \end{equation}
\end{lemma}

\begin{assumption}\label{app_ass_E5}
For every fixed $R\ge1$, conditionally on the data,
$$
        \left(\hat v_{n1}^{\top}e,\dots,\hat v_{nR}^{\top}e,\,
        \frac{1}{\sqrt 2}\sum_{\ell>J_n}^{n_2}\hat\lambda_{n\ell}\{(\hat v_{n\ell}^{\top}e)^2-1\}\right)\overset{D^*}{\to}(G_1,\dots,G_R,\rho G_0),
$$
where $G_0,G_1,\dots,G_R$ are independent standard Gaussian random variables. 
\end{assumption}

Assumption~\ref{app_ass_E5} is required for Rademacher multipliers, while the replacement of $(\hat v_{n\ell}^{\top}e)$ by independent standard Gaussian variables is exact only for Gaussian multipliers. For Rademacher multipliers, a standard sufficient condition for the finite-dimensional part is that, for $R\ge1$, it holds
$$
        \max_{1\le i\le n_2}\sum_{\ell=1}^{R}\hat v_{n\ell,i}^2\overset{p}{\to}0.
$$
For the joint convergence with the quadratic form, let
$$
        \mathbf A_{n,J_n}=\sum_{\ell>J_n}^{n_2}\hat\mu_{n\ell}\hat v_{n\ell}\hat v_{n\ell}^{\top},\quad \mathbf A_{n,J_n}^c=\mathbf A_{n,J_n}-\operatorname{diag}(\mathbf A_{n,J_n}).
$$
It is sufficient to verify
$$
        \frac{\|\mathbf A_{n,J_n}^c\|_F^2}{\hat s_n^2}
        \overset{p}{\to}\rho^2,
        \quad
        \frac{\|\mathbf A_{n,J_n}^c\|_{\mathrm{op}}}
             {\|\mathbf A_{n,J_n}^c\|_F}
        \overset{p}{\to}0. 
$$
Then the joint convergence holds by verifying the bounded-degree low-influence invariance principle \citep{MosselODonnellOleszkiewicz2010} and the limiting theorem for vectors of multiple Gaussian integrals \citep{PeccatiTudor2005}.

The following theorem establishes the validity results of wild bootstrap with varying kernel. 

\begin{theorem}\label{app_th_02}
Suppose Assumptions~\ref{ass_01}, \ref{app_ass_E1}--\ref{app_ass_E5} hold and $E(k(X,X))<\infty$. Under the null $H_0$, if $n_1,n_2\to\infty$ and $n_2=o(r_{n_1}^{-2})$, we have
$$
        \sup_{t\in\mbR}
        \left|P^*\left(\frac{n_2 T_k^b}{\sqrt 2\hat s_n}\le t\right)-P\left(\frac{n_2\tilde T_k}{\sqrt 2 s_n}\le t\right)\right|\overset{p}{\to}0 ,
$$
where $P^*(\cdot)=P(\cdot| \mcD_{n})$.
\end{theorem}

To end the discussion, we remark that the proofs of Theorem~\ref{th_power} in Section~\ref{app_sec_026} and Theorem~\ref{th_04} in Section~\ref{app_sec_025} also apply to varying kernel settings.

\subsection{Technical Proofs}
\subsubsection{Proof of Lemma~\ref{app_lemma_E1}}
\begin{proof}
In fact, Fatou's lemma gives that
$$
     \sum_{r=1}^{\infty}\lambda_r^2
    =\sum_{r=1}^{\infty}\liminf_{n\to\infty}\lambda_{nr}^2\leq\liminf_{n\to\infty}\sum_{r=1}^{\infty}\lambda_{nr}^2=1. 
$$
Take a strictly increasing sequence of positive integers $M_j\uparrow\infty$ such that 
$$
    \sum_{r>M_j}^{\infty}\lambda_{r}^2\leq \frac{1}{j}. 
$$
For every fixed $M_j$, Assumption~\ref{app_ass_E1} yields that 
$$\sum_{r=1}^{M_j}\lambda_{nr}^2\to\sum_{r=1}^{M_j} \lambda_r^2.$$
Therefore, we choose a strictly increasing sequence of
positive integers $N_j\uparrow\infty$, with $N_j\geq j$, such that for $n_2\geq N_j$, 
$$
\bigg|\sum_{r=1}^{M_j}\lambda_{nr}^2-\sum_{r=1}^{M_j} \lambda_r^2\bigg|\leq \frac{1}{j}. 
$$
For $n\geq N_1$, define
\[
    j(n):=\max\{j\geq1:N_j\leq n\},\quad R_n:=M_{j(n)},
\]
and define $R_n=1$ for $n<N_1$. Since $N_j\uparrow\infty$, we have $j(n)\to\infty$, and then $R_n=M_{j(n)}\to\infty$. Moreover,
$$
\bigg|\sum_{r=1}^{R_n}\lambda_{nr}^2-\sum_{r=1}^\infty\lambda_{nr}^2\bigg|
\leq \bigg|\sum_{r=1}^{M_{j(n)}}\lambda_{nr}^2-\sum_{r=1}^{M_{j(n)}}\lambda_{r}^2\bigg|+\sum_{r>M_{j(n)}}^{\infty}\lambda_{r}^2\leq\frac{2}{j(n)},
$$
which converges to $0$. This completes the proof of \eqref{app_eq_D01}. 
\end{proof}

\subsubsection{Proof of Theorem~\ref{app_th_01}}
\begin{proof}

In the proof of Theorem~\ref{th_02} in Section~\ref{app_sec_023}, we have proved that 
$$n_2T_k=n_2\tilde T_k+o_p(1).$$
Since $s_n=\Omega(1)$ in high- and fixed-dimensional settings, it sufficient to show that, as $n_2\to\infty$, 
$$
        \frac{n_2\tilde T_k}{\sqrt 2 s_n}\overset{d}{\to}\frac{1}{\sqrt 2}\sum_{r=1}^\infty \lambda_r(Z_r^2-1)+\rho Z_0. 
$$

The $L_2$-spectral expansion gives, for $i\ne j$,
$$
        h(Z_{2,i},Z_{2,j})=\sum_{r=1}^\infty\kappa_{nr}\xi_{n,i,r}\xi_{n,j,r},
$$
where the equality holds in $L_2$. Hence
$$
\begin{aligned}
        \frac{n_2\tilde T_k}{\sqrt 2 s_n}=\frac{1}{\sqrt 2 s_n(n_2-1)}\sum_{1\leq i\neq j\leq n_2}h(Z_{2,i},Z_{2,j})                                     
        =:\frac{n_2}{n_2-1}\frac{1}{\sqrt 2}\sum_{r=1}^\infty\lambda_{nr}A_{nr},
\end{aligned}
$$
where
$$
        A_{nr}=\frac{1}{n_2}\left\{\left(\sum_{i=1}^{n_2}\xi_{n,i,r}\right)^2-\sum_{i=1}^{n_2}\xi_{n,i,r}^2\right\}.
$$
Fix $R\geq1$. By taking $(a,b)=(e_r,0)$ and $(a,b)=(e_r+e_s,0)$ in Assumption~\ref{app_ass_E3} and some calculations, for every fixed $R\geq1$, it holds that
\begin{equation}\label{app_E_eq_04}
    \max_{1\le r,s\le R}\left|\frac{1}{n_2}\sum_{i=1}^{n_2}\xi_{n,i,r}\xi_{n,i,s}
    -\mathbbm 1(r=s)\right|\overset{p}{\to}0,
\end{equation} 
where $e_r$ denotes an $R$-dimensional vector with the $r$-th element being $1$ and the others being $0$. The martingale central limit theorem along with some calculations implies
$$
    \left(n_2^{-1/2}\sum_{i=1}^{n_2}\xi_{n,i,1},\dots,n_2^{-1/2}\sum_{i=1}^{n_2}\xi_{n,i,R},\sum_{j=1}^{n_2}G_{n,j}\right)\overset{d}{\to}(Z_1,\dots,Z_R,\rho Z_0),
$$
where $Z_0,Z_1,\dots,Z_R$ are independent standard Gaussian variables and $\rho^2=1-\sum_{r=1}^{\infty}\lambda_r^2$. Combining this with \eqref{app_E_eq_04},  Assumption~\ref{app_ass_E1} and Slutsky's theorem yields
$$
        \frac{n_2}{n_2-1}\frac{1}{\sqrt 2}\sum_{r=1}^R\lambda_{nr}A_{nr}
        +\sum_{j=1}^{n_2}G_{n,j}\overset{d}{\to}\frac{1}{\sqrt 2}\sum_{r=1}^R\lambda_r(Z_r^2-1)+\rho Z_0 .
$$
Here, the tail martingale is exactly the spectral tail, i.e., 
$$
    \sum_{j=1}^{n_2}G_{n,j}=\frac{n_2}{n_2-1}\frac{1}{\sqrt2}\sum_{r>R_n}\lambda_{nr}A_{nr}.
$$

It remains to show that the intermediate block $R<r\le R_n$ is negligible as $R\to\infty$. For $r,s\ge1$,
$$
        E(A_{nr}A_{ns})=\frac{2(n_2-1)}{n_2}\mathbbm 1(r=s).
$$
Therefore, by the orthogonality of eigenfunctions,
$$
\begin{aligned}
        E\left[\left\{\frac{n_2}{n_2-1}\frac{1}{\sqrt 2}\sum_{r=R+1}^{R_n}\lambda_{nr}A_{nr}\right\}^2\right]&=\frac{n_2}{n_2-1}\sum_{r=R+1}^{R_n}\lambda_{nr}^2 .
\end{aligned}
$$
For every fixed $R$, Assumption~\ref{app_ass_E1} with \eqref{app_eq_D01} gives
$$
\sum_{r=R+1}^{R_n}\lambda_{nr}^2\to \sum_{r>R}\lambda_r^2,
$$
which is followed by 
$$
        \limsup_{n\to\infty}E\left[\left\{\frac{n_2}{n_2-1}\frac{1}{\sqrt 2}\sum_{r=R+1}^{R_n}\lambda_{nr}A_{nr}\right\}^2\right]=\sum_{r>R}\lambda_r^2 .
$$
Therefore, for every $\varepsilon>0$, Chebyshev's inequality implies
    $$
    \limsup_{n\to\infty}P\bigg(\bigg|\frac{n_2}{n_2-1}\frac{1}{\sqrt 2}\sum_{r=R+1}^{R_n}\lambda_{nr}A_{nr}\bigg|>\varepsilon\bigg)\leq
    \frac{1}{\varepsilon^2}\sum_{r>R}\lambda_r^2.
    $$
Since $\sum_{r=1}^{\infty}\lambda_r^2<\infty$, it follows that
    $$
    \lim_{R\to\infty}\limsup_{n\to\infty}P\bigg(\bigg|\frac{n_2}{n_2-1}\frac{1}{\sqrt 2}\sum_{r=R+1}^{R_n}\lambda_{nr}A_{nr}\bigg|>\varepsilon\bigg)=0.
    $$

Finally, as $R\to\infty$,
\[
    E\bigg[\bigg\{\frac1{\sqrt2}\sum_{r>R}\lambda_r(Z_r^2-1)\bigg\}^2\bigg]=\sum_{r>R}\lambda_r^2\to0.
\]
Markov's inequality along with some calculations yields that 
$$
        \frac{n_2\tilde T_k}{\sqrt 2 s_n}
        \overset{d}{\to}
        \frac{1}{\sqrt 2}\sum_{r=1}^\infty \lambda_r(Z_r^2-1)+\rho Z_0,
$$
where the series on the right-hand side converges in $L_2$ since $\sum_r\lambda_r^2<\infty$. It completes the proof.
\end{proof}

\subsubsection{Proof of Lemma~\ref{app_lemma_E2}}
\begin{proof}
Take a strictly increasing sequence of positive integers $M_j\uparrow\infty$ such that
\[
    \sum_{\ell>M_j}^{\infty}\lambda_\ell^2\leq\frac1j.
\]
For every fixed $M_j$, the continuous mapping theorem yields that
\[
    \sum_{\ell=1}^{M_j}\hat\lambda_{n\ell}^2\overset{p}{\to}\sum_{\ell=1}^{M_j}\lambda_\ell^2.
\]
Therefore, we can choose a strictly increasing sequence of positive integers $N_j\uparrow\infty$, with $N_j\geq j$, such that, for every $n\geq N_j$,
\[
    P\bigg(\bigg|\sum_{\ell=1}^{M_j}\hat\lambda_{n\ell}^2-\sum_{\ell=1}^{M_j}\lambda_\ell^2\bigg|>\frac1j\bigg)\leq\frac1j,
\]
and such that $n_2\geq M_j$ whenever $n\geq N_j$.

For $n\geq N_1$, define
\[
    j(n):=\max\{j\geq1:N_j\leq n\}, \quad J_n:=M_{j(n)},
\]
and define $J_n=1$ for $n<N_1$. Since $N_j\uparrow\infty$, we have $j(n)\to\infty$, and hence $J_n=M_{j(n)}\to\infty$.  Moreover, $J_n\leq n_2$ for every $n\geq N_1$. It follows that
\[
 P\bigg(\bigg|\sum_{\ell=1}^{J_n}\widehat\lambda_{n\ell}^2-\sum_{\ell=1}^{\infty}\lambda_\ell^2\bigg|>\frac{2}{j(n)}\bigg)
 \leq P\bigg(\bigg|\sum_{\ell=1}^{M_{j(n)}}\widehat\lambda_{n\ell}^2-\sum_{\ell=1}^{M_{j(n)}}\lambda_\ell^2\bigg|>\frac{1}{j(n)}\bigg)
     \leq\frac{1}{j(n)}.
\]
Since $j(n)\to\infty$, it follows that
\[
    \sum_{\ell=1}^{J_n}\hat\lambda_{n\ell}^2\overset{p}{\to}\sum_{\ell=1}^{\infty}\lambda_\ell^2.
\]

Finally, by the identity
\[
    \sum_{\ell=1}^{n_2}\hat\lambda_{n\ell}^2=\frac{n_2}{n_2-1},
\]
we have
\[
    \sum_{\ell>J_n}^{n_2}\widehat\lambda_{n\ell}^2=
    \frac{n_2}{n_2-1}
    -\sum_{\ell=1}^{J_n}\widehat\lambda_{n\ell}^2\overset{p}{\to}
    1-\sum_{\ell=1}^{\infty}\lambda_\ell^2
    =\rho^2.
\]
This completes the proof.
\end{proof}

\subsubsection{Proof of Theorem~\ref{app_th_02}}
\begin{proof}
In the proof of Theorem~\ref{th_03} in Section~\ref{app_sec_024}, we have proved that 
$$
E^*\{(T_k^b-\tilde T_k^b)^2\} = o_p(n_2^{-2}). 
$$
Since $\tilde T_k^b\overset{d}{=}\tilde T_k^*$ and $\hat s_n = \Omega_p(1)$, conditional Markov's inequality along with some elementary derivations gives that
$$
        \sup_{t\in\mbR}
        \left|P^*\left(\frac{n_2 T_k^b}{\sqrt 2\hat s_n}\le t\right)-P^*\left(\frac{n_2 \tilde T_k^*}{\sqrt 2\hat s_n}\le t\right)\right|\overset{p}{\to}0 ,
$$
By the triangle inequality, it is sufficient to prove 
$$
        \sup_{t\in\mbR}
        \left|P^*\left(\frac{n_2\tilde T_k^*}{\sqrt 2\hat s_n}\le t\right)-P\left(\frac{n_2\tilde T_k}{\sqrt 2 s_n}\le t\right)\right|\overset{p}{\to}0.  
$$

Since $\operatorname{tr}(\mathbf H_n)=0$, $\sum_{\ell=1}^{n_2}\hat\mu_{n\ell}=0$. Therefore, conditionally on the data,
\begin{equation}\label{app_eq_D03}
    \frac{n_2\tilde T_k^*}{\sqrt 2\hat s_n}=\frac{1}{\sqrt 2\hat s_n}e^\top\mathbf H_n e=\frac{1}{\sqrt 2}\sum_{\ell=1}^{n_2}\hat\lambda_{n\ell}\{(\hat v_{n\ell}^{\top}e)^2-1\}.
\end{equation}
Fix $R\ge1$. By Assumptions~\ref{app_ass_E4} and \ref{app_ass_E5},
\begin{equation}\label{app_eq_D04}
        \frac{1}{\sqrt 2}\sum_{\ell=1}^{R}\hat\lambda_{n\ell}\{(\hat v_{n\ell}^{\top}e)^2-1\}
        +\frac{1}{\sqrt 2}\sum_{\ell>J_n}^{n_2}\hat\lambda_{n\ell}\{(\hat v_{n\ell}^{\top}e)^2-1\}\overset{D^*}{\to}\frac{1}{\sqrt 2}\sum_{\ell=1}^{R}\lambda_\ell(G_\ell^2-1)+\rho G_0.
\end{equation}

It remains to control the intermediate block $R<\ell\le J_n$.  Let
$$
        \mathbf A_{n,R}=\sum_{\ell=R+1}^{J_n}\hat\mu_{n\ell}\hat v_{n\ell}\hat v_{n\ell}^{\top}.
$$
and 
$$
        M_{n,R}^*=\frac{1}{\sqrt 2\hat s_n}\{e^\top \mathbf A_{n,R}e-\operatorname{tr}(\mathbf A_{n,R})\}.
$$
A direct calculation shows that
$$
        \Var^*\{e^\top \mathbf A_{n,R}e-\operatorname{tr}(\mathbf A_{n,R})\}=2\sum_{i\ne j}(\mathbf A_{n,R})_{ij}^2\le2\|\mathbf A_{n,R}\|_F^2.
$$
which is followed by 
$$
        \Var^*(M_{n,R}^*)\le\sum_{\ell=R+1}^{J_n}\hat\lambda_{n\ell}^2 .
$$
By Lemma~\ref{app_lemma_E2},
$$
        \lim_{R\to\infty}\limsup_{n\to\infty}P\left\{\Var^*(M_{n,R}^*)>\eta\right\}=0\quad\text{for every }\eta>0.
$$
Markov's inequality then yields for every $\epsilon,\eta>0$,
\[
    \lim_{R\to\infty}\limsup_{n\to\infty}
    P\{P^*(|M_{n,R}^*|>\epsilon)>\eta\}=0. 
\] 
Combining it with \eqref{app_eq_D03} and \eqref{app_eq_D04} yields that
$$
        \frac{n_2\tilde T_k^*}{\sqrt 2\hat s_n}
        \overset{D^*}{\to}L:=\frac{1}{\sqrt 2}\sum_{\ell=1}^{\infty}\lambda_\ell(G_\ell^2-1)+\rho G_0.
$$

Note that the distribution function of $L$ is continuous. Indeed, if $\rho>0$, this follows from convolution with a nondegenerate Gaussian distribution. If $\rho=0$, then at least one $\lambda_\ell$ is nonzero and $L$ is the sum of an absolutely continuous random variable and an independent remainder.
Therefore, by Lemma~2.11 in \cite{vanDerVaart1998} and conditional weak convergence,
$$
        \sup_{t\in\mbR}\left|P^*\left(\frac{n_2\tilde T_k^*}{\sqrt 2\hat s_n}\le t\right)-P(L\le t)\right|\overset{p}{\to}0 .
$$
By Theorem~\ref{app_th_01} and the same continuity argument,
$$
        \sup_{t\in\mbR}\left|P\left(\frac{n_2\tilde T_k}{\sqrt 2 s_n}\le t\right)
        -\mathbb P(L\le t)\right|\to0.
$$
Combining the last two displays with the triangle inequality proves the theorem.  
\end{proof}

\section{Additional Numerical Results}\label{app_sec_05}
\subsection{Rejection Rates with Correlated Predictors}\label{app_sec_simu}
In this subsection, we report the results in Simulations~1--3 in the main text with correlated predictors. As illustrated in Figures~\ref{figure_1correlated}--\ref{figure_3correlated}, the rejection rates with correlated predictors perform similarly to those with the independent predictors. 
In general, our method appears to control the Type-I error well and shows a good power performance across the settings. 
The HCZ test may suffer an inflated Type-I error rate with a small sample size and loss of power with low-dimensional predictors. 
Moreover, although the AN test exhibits the highest power in low-dimensional settings as shown in Figure~\ref{figure_2correlated}, our method performs close to the AN test with a large sample size and shows a superior to the HCZ and RP test.

\begin{figure}[htbp]\centering
\includegraphics[width=1\textwidth]{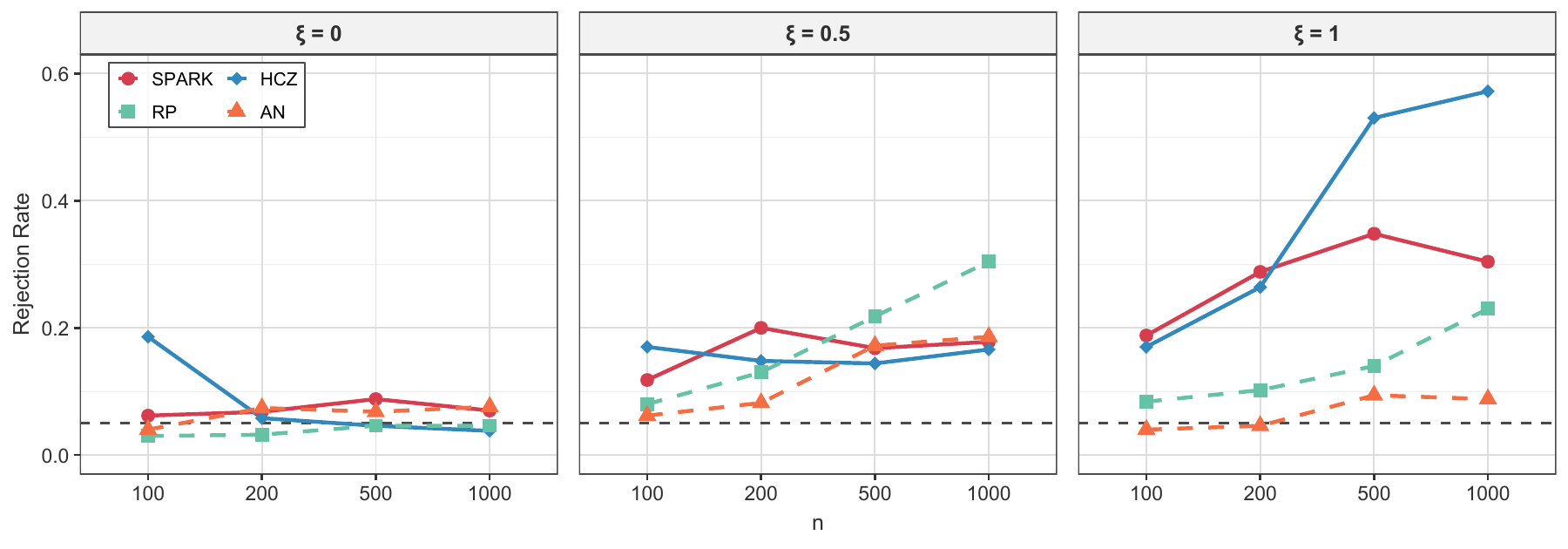}
\caption{Empirical Type-I error and power under Simulation~1 with correlated predictors. The horizontal axis denotes the sample size $n$, and the vertical axis shows the rejection rate. The black dashed line indicates the nominal significance level of $0.05$. }\label{figure_1correlated}
\end{figure}

\begin{figure}[htbp]\centering
\includegraphics[width=1\textwidth]{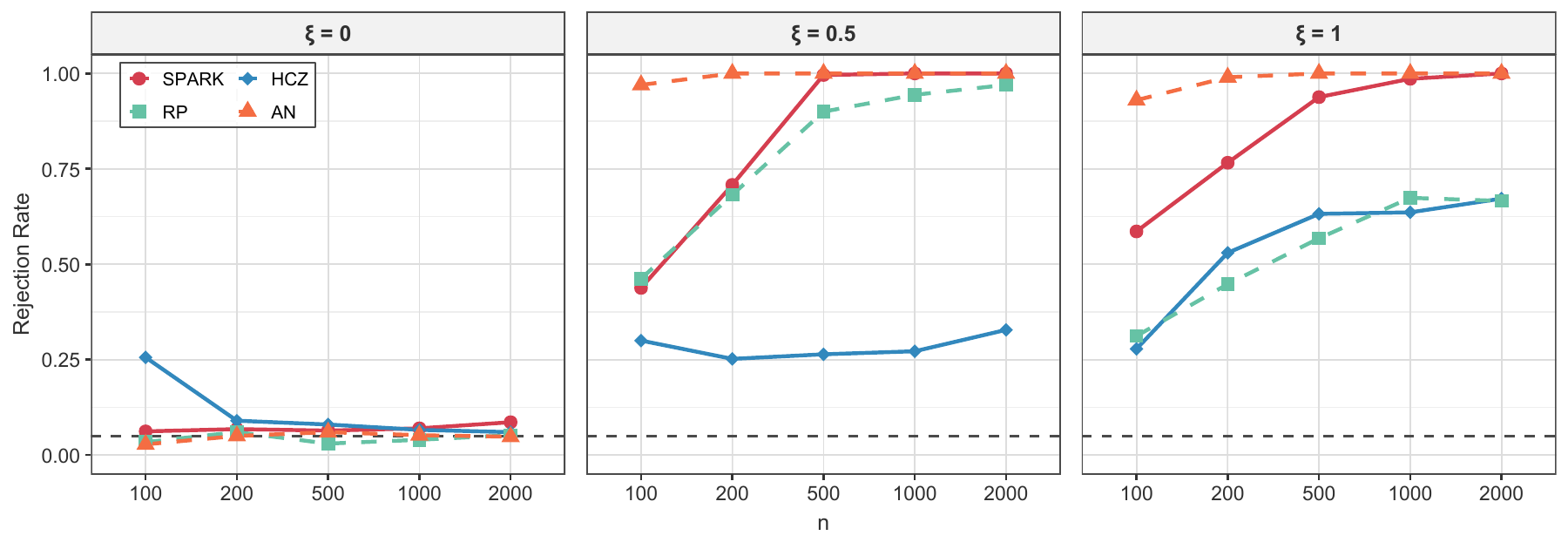}
\caption{Empirical Type-I error and power under Simulation~2 with correlated predictors. The horizontal axis denotes the sample size $n$, and the vertical axis shows the rejection rate. The black dashed line indicates the nominal significance level of $0.05$. }\label{figure_2correlated}
\end{figure}

\begin{figure}[htbp]\centering
\includegraphics[width=1\textwidth]{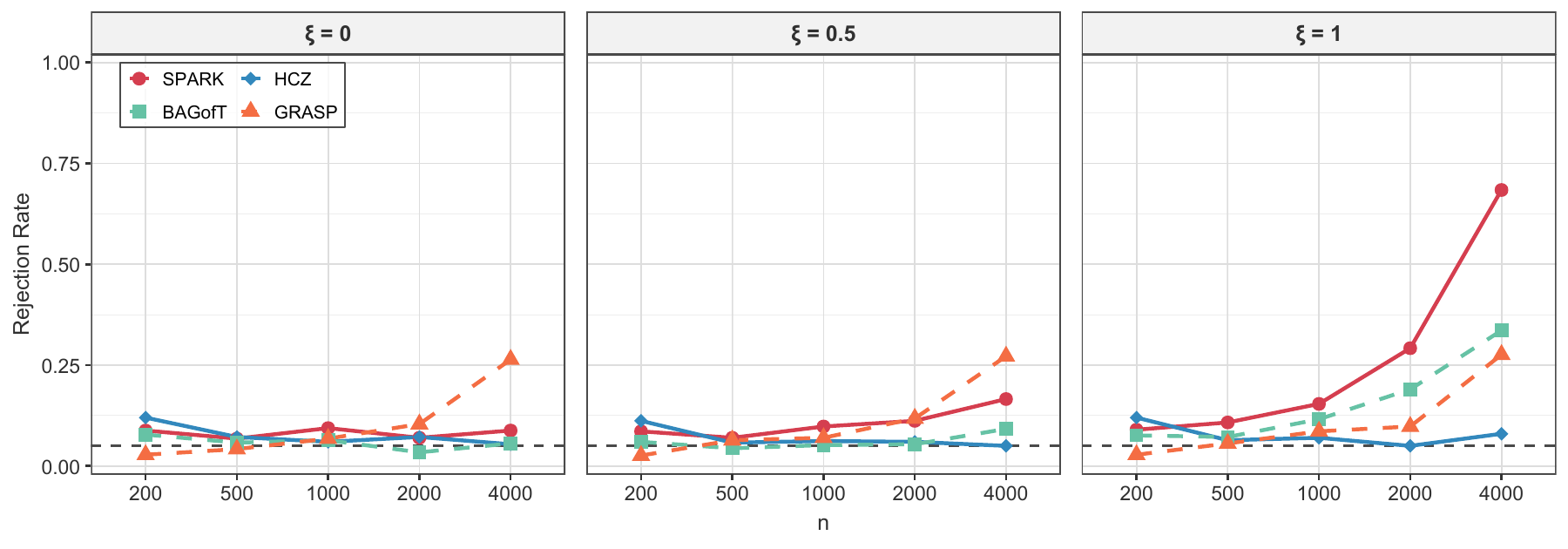}
\caption{Empirical Type-I error and power under Simulation~3 with correlated predictors. The horizontal axis denotes the sample size $n$, and the vertical axis shows the rejection rate. The black dashed line indicates the nominal significance level of $0.05$. }\label{figure_3correlated}
\end{figure}

\subsection{Cortisol Stress Reactivity Dataset}\label{app_sec_simu_2}
In this subsection, we analyze cortisol stress reactivity dataset, which is a well-established high-dimensional dataset widely used in mediation analysis to examine the role of DNA methylation in mediating the relationship between childhood trauma and cortisol stress reactivity \citep{Houtepen2016GenomewideDM,Guo2022HighDimensionalMA,Hehexu2025ADA}. 
The dataset consists of $n=85$ observations and $385,882$ DNA methylation loci, childhood trauma status, cortisol stress reactivity, six immune cell proportions, and confounding variables such as age and sex. The dataset is publicly available at \url{https://www.ebi.ac.uk/biostudies/arrayexpress/studies/EGEOD-77445}. 

Following \cite{He2025AGA}, we investigate the association between DNA methylation loci and the CD8 T cell proportion, a key immune component. 
We first using a marginal screening procedure \citep{Fan2008SureIS} to reduce the number of features to $p = 1000$. We then apply our method and HCZ to assess the performance of the six methods described in Section~\ref{sec_52}, LASSO, SCAD, SVR, RF, XGBoost, and FNN in predicting CD8 T cell proportions. 

Due to the limited sample size, we employ a fixed splitting ratio with $n_1=20$ for training and $n_2=65$ for testing for our method. We adopt a fixed data split, allocating $n_1 = 70$ and $n_2 = 15$ for HCZ as they used.
Figure~\ref{cortisol_boxplots_combined} presents the $p$-values of our test and HCZ and the MSEs obtained over $100$ replications, where each $p$-value is obtained using $10$ multiple splittings with the Cauchy combination. 

\begin{figure}[htbp]\centering
\includegraphics[width=1\textwidth]{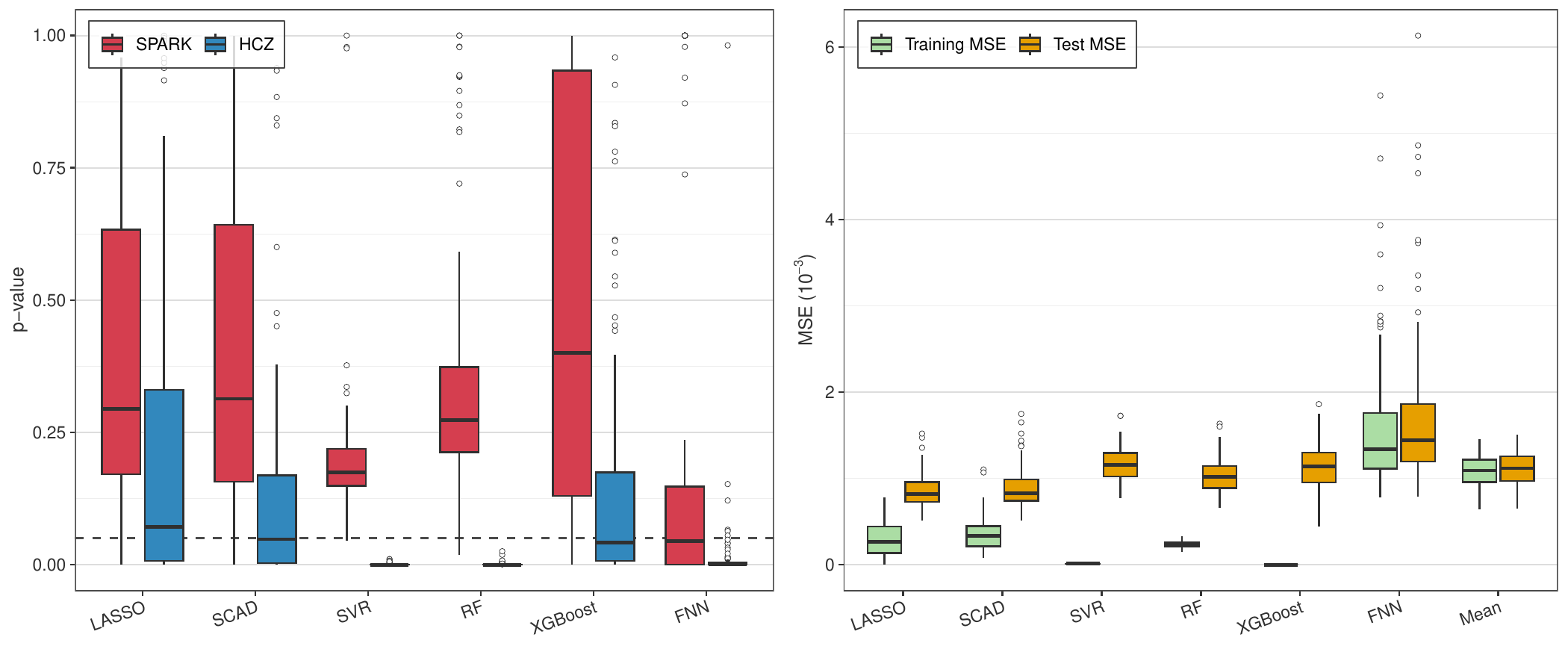}
\caption{
Boxplots of the $p$-values and MSEs for cortisol stress reactivity dataset. 
The horizontal axis indicates different learning procedure. The black dashed line indicates the nominal significance level of $0.05$. MSEs are estimated using training size $n_1=45$ and testing size $n_2=40$. }\label{cortisol_boxplots_combined}
\end{figure}


Our results show that FNN may not fit this dataset well, while there is not enough evidence to suggest that the other methods fail to converge to the true regression function. Therefore, a linear model, as used in \cite{Guo2022HighDimensionalMA}, could be a reasonable assumption for the relationship between DNA methylation and CD8 T-cell proportion. 
However, the results of HCZ suggest that only LASSO, SCAD, and XGBoost may converge, with strong evidence against the null hypothesis for the other three methods. 
Note that only the test MSEs of FNN are significantly larger than those obtained by using the mean of the training responses, thus our findings are more in line with the MSE results. 
Moreover, given the inflated Type I error of HCZ in small-sample, high-dimensional settings shown in Simulation~1, our results are likely to be more reliable. 

\bibliographystyle{apalike}
\bibliography{reference}
\end{appendices}

\end{document}